\documentclass[]{article}

\usepackage{graphics} % for pdf, bitmapped graphics files
\usepackage{amsmath} % assumes amsmath package installed
\usepackage{amsthm}
\usepackage{ulem}

\usepackage{amssymb}  % assumes amsmath package installed
\usepackage{color}
\usepackage{multirow}
\usepackage{bm}
\usepackage{cite}
\usepackage{amsmath,amssymb,dsfont}
\usepackage{algorithmic}
\usepackage{graphicx}
\usepackage{graphics} % for pdf, bitmapped graphics files
\usepackage{epsfig} % for postscript graphics files
\usepackage{algorithm}

\usepackage{algorithmic}

\usepackage{graphicx}           	% Include this line if your 
\usepackage{epstopdf}           	%to include eps figure in the latex file
\usepackage{subfloat}				%to include subfigures
\usepackage{subfigure}
\usepackage{url}
\usepackage{tikz}
\usepackage{verbatim}
\usepackage{xcolor}
\usepackage{soul,color}

\newcommand{\V}{\mathcal{V}}
\newcommand{\E}{\mathcal{E}}

\newtheorem{definition}{Definition}

\newtheorem{corollary}{Corollary}
\newtheorem{theorem}{Theorem}
\newtheorem{lemma}{Lemma}
\newtheorem{assumption}{Assumption}
\newtheorem{proposition}{Proposition}
\newtheorem{remark}{Remark}

\newtheorem{policy}{Policy}

\newcommand{\argmin}[1]{\underset{#1}{\mathrm{argmin\,}}}
\newcommand{\vx}{\boldsymbol{x}}

\newcommand{\vu}{\boldsymbol{u}}
\newcommand{\vy}{\boldsymbol{y}}

\renewcommand{\eqref}[1]{Eq.~(\ref{#1})}

\newcommand\diag{\operatorname{diag}}

\newcommand{\veta}{\boldsymbol{\eta}}

\newcommand{\real}{\mathbb{R}}

\newcommand{\integer}{\mathbb{Z}}
\newcommand{\integernonnegative}{\mathbb{Z}_{\ge 0}}

\newcommand{\Exp}{\mathds{E}} % expectation
\def\Exp{\mathbb{E}}

\title{\LARGE \bf
Title
}

	\title{Optimal Policy Design for Repeated Decision-Making under Social Influence}
		\author{Chiara Ravazzi, Valentina Breschi,  Paolo Frasca, \\ Fabrizio Dabbene, Mara Tanelli}
\begin{document}

\maketitle
\thispagestyle{empty}
\pagestyle{empty}

%%%%%%%%%%%%%%%%%%%%%%%%%%%%%%%%%%%%%%%%%%%%%%%%%%%%%%%%%%%%%%%%%%%%%%%%%%%%%%%%
\begin{abstract}
In this paper, we present a novel model to characterize individual tendencies in repeated decision-making scenarios, with the goal of designing model-based control strategies that promote virtuous choices amidst social and external influences. Our approach builds on the classical Friedkin and Johnsen model of social influence, extending it to include random factors (e.g., inherent variability in individual needs) and controllable external inputs.
We explicitly account for the temporal separation between two processes that shape opinion dynamics: individual decision-making and social imitation. While individual decisions occur at regular, frequent intervals, the influence of social imitation unfolds over longer periods. 
The inclusion of random factors
naturally leads to dynamics that do not converge in the classical sense. However, under specific conditions, we prove that opinions exhibit ergodic behavior.
Building on this result, we propose a constrained asymptotic optimal control problem designed to foster, on average, social acceptance of a target action within a network. To address the transient dynamics of opinions, we reformulate this problem within a Model Predictive Control (MPC) framework. Simulations highlight the significance of accounting for these transient effects in steering individuals toward virtuous choices while managing policy costs.
\end{abstract}

\section{Introduction}
\label{sec:intro}
In our daily routines, we often make repeated decisions. For example, each day we choose the mode of transportation for our commute, considering various factors such as our schedule, weather conditions, economic costs, and even ethical or social norms that influence our preference for specific mobility options.

From this simple example, it is already clear that these decisions are  heavily influenced by both the social environment one is immersed in and other external factors \cite{friedkin1986formal}, as well as by our habits and the resistance to change them. Indeed, individuals often demonstrate stubborn behaviors, being hesitant to change their habits (even when it could benefit them in the long run) due to either rational or psychological factors \cite{friedkin1990social}. 

Changes in one's habits can nonetheless be fostered by tailor-made policies that encourage specific virtuous actions, e.g., the daily adoption of sustainable mobility services like bike sharing. Given the interplay of different factors in shaping individual recurrent choices, these policies should account for $(i)$ the characteristics of individuals, $(ii)$ the complex relationships between them, as well as $(iii)$ the consequences of past policy decisions (as discussed in \cite{breschi2022driving,breschi2022fostering}). In a related line of research, some recent works started investigating the influence of recommendations on the dynamics within social networks. Based on \cite{rossi2022closed}, where a feedback interaction between a recommendation system and a single user is introduced, \cite{castro2018opinion,goyal2019maintaining} integrate the dynamics of social interactions and the recommendation system within influence networks. Meanwhile, the paper \cite{sprenger2024control} leverages an extended Friedkin and Johnsen model \cite{friedkin1990social} to introduce a Model Predictive Control (MPC) strategy for user engagement maximization, hence designing a (non-personalized) recommendation policy that accounts for the evolution of opinions over time.    
\paragraph*{Contribution} In this work, our ultimate goal is to design effective nudge policies that favor virtuous actions within repetitive decision-making scenarios that account for the different players shaping the dynamics of opinions. To this end, we propose a new opinion dynamics model based on the well-established Friedkin and Johnsen model \cite{friedkin1990social}. Taking advantage of the ability of the latter to capture the intricate interplay between individual behaviors and social influences, our model also incorporates \textit{random factors} and \textit{external} (controlled) policies that can shape opinion dynamics. The former represents the variability likely to affect individual inclination over repetitive choices, as well as the random nature of social acceptance toward specific actions that a policymaker seeks to promote. The latter allows us to incorporate the effect of external policies in shaping personal choices. In the absence of external control actions, the proposed dynamics does not converge, but persistently oscillates. Nonetheless, under suitable assumptions, we prove that the resulting stochastic process modeling individual opinions converges almost surely to a final limit and that the system is ergodic. In turn, this implies that, over a sufficiently long period, the time averages of inclinations converge to their expected values.

With this in mind, we propose two optimization-based control frameworks to foster the propagation and acceptance of desired actions within the network on average while containing policy efforts. Specifically, we initially design control actions only looking at the asymptotic behavior of average individual inclinations, then considering also transients by formulating an MPC problem to foster virtuous actions in repetitive settings. Compared to \cite{sprenger2024control}, in formulating the policy design problem, the recommender is not intended as a new node in the network able to influence at each time individuals. Instead, we intend the policymaker as an external influence able to act directly on the bias of individual inclination. This choice allows us $(i)$ to design \textit{personalized} policies for each individual, and $(ii)$ to explicitly account for the fact that nudging/incentive policies often do not act directly on one's inclination, but rather aim at tearing down specific barriers that prevent a certain virtuous choice from being made (e.g., fidelity discounts for bike-sharing users). Furthermore, unlike prior literature, our policy design approach does not entail observing opinions reflecting inclinations toward acceptance. Instead, we are restricted to observing only their (binary) realizations. This choice, together with the introduction of a budget constraint, allows us to formulate the decision-making problem considering two core practical aspects in policy design, e.g., limited resources and access to limited information on individual opinions.  

The impact of personalization is evaluated in our numerical example by comparing the proposed policies with a control action obtained by uniformly dividing the available budget across individuals, showcasing the benefit of tailor-made policy actions.
\paragraph*{Outline} We introduce the proposed opinion dynamics model in Section~\ref{sec:model}, characterizing its properties in the absence of control and presence of a persistent constant control action in Section~\ref{sec:ergodicity}. Two constant control policies, respectively obtained by uniformly dividing the available budget across individuals and optimizing an acceptance/cost trade-off by looking at the asymptotic behavior of opinions only, are introduced in Section~\ref{sec:CCP}. These policies will serve as benchmarks for comparison with the model predictive controller presented in Section~\ref{sec:control}, where we consider both the cases where an oracle provides prior information on the average inclinations and such information is estimated from the available observations. The different strategies are then analyzed in Section~\ref{sec:numerical} on a numerical example. We conclude the paper with some remarks and discussion on future developments. 
\paragraph*{Notation and useful facts} The sets of real and nonnegative integer numbers are denoted by $\real$ and $\integernonnegative$, respectively. Column vectors are indicated in boldface, while matrices are boldfaced capital letters. Unitary and zero (column) vectors are accordingly denoted as $\boldsymbol{1}$ and $\boldsymbol{0}$, respectively. A matrix $\boldsymbol{M}\in\real^{n\times n}$ is said row-stochastic if its entries are all nonnegative and $\boldsymbol{M}\boldsymbol{1}=\boldsymbol{1}$, while it is Schur stable if all its eigenvalues lay inside the unit circle. Moreover, $[\boldsymbol{M}]$ denotes a diagonal matrix whose entries are equal to the diagonal of $\boldsymbol{M}$. Given a random vector $\boldsymbol{v}$, $\mathbb{E}[\boldsymbol{v}]$ denotes its expected value.

Given a stochastic process \cite{doob1962stochastic} $\{\vx(t)\}_{\integernonnegative}$, its time-average (also Ces\'aro or Polyak average in some contexts~\cite{polyak1992acceleration}) is given by
\begin{equation}\label{eq:time_average}
 \overline{\vx}(t)=\frac{1}{t}\sum_{\ell=0}^{t-1}\vx({\ell}).   
\end{equation}
By relying on this definition, we can then define the ergodicity of a random process as follows.
\begin{definition}[Ergodicity of a random process]\label{def:ergo}{We say that a stochastic process $\{\vx(t)\}_{t\in\integernonnegative}$ is ergodic if there exists a random variable $\vx_{\infty} $ such that 
\begin{equation}\label{eq:ergodicity}
\lim_{t\to \infty}\overline{\vx}(t)=\Exp[\vx_{\infty}],
\end{equation} 
almost surely}\end{definition}
A process is instead mean-square ergodic when 
\begin{equation}
    \lim_{k \to \infty} \Exp\left[\left\|\overline{\vx}(t)-\Exp[\vx_{\infty}]
\right\|_2^2\right] = 0.
\end{equation}
In our analysis, we focus on the notion of ergodicity provided in Definition~\ref{def:ergo}. Nonetheless, note that under specific assumptions, mean-square convergence can be derived from almost-sure convergence. For instance, this implication is true for a uniformly bounded sequence of random variables, by the Dominated Convergence Theorem~\cite{borkar1995probability}.  

\section{A new model for repetitive decision-making scenarios} \label{sec:model}
We consider a social network represented by a directed graph\footnote{To avoid trivial cases, we implicitly assume that graphs comprise at least three nodes.} $ \mathcal{G} =(\mathcal{V,E},\boldsymbol{P})$, consisting of a set $\mathcal{V} $ of agents and a set of edges $\mathcal{E}\in \mathcal{V\times V}$ describing the (directed) influences among them, i.e., if $(v,w)\in\E$, we will say that $v$ is potentially influenced by node $w\in\mathcal{V}$. The intensities of their potential interactions are encoded in a weight matrix $\boldsymbol P\in\mathbb{R}^{|\mathcal{V}|\times|\mathcal{V}|}$, that we assume to be row-stochastic and nonnegative, while the set of agents influencing node $v\in\mathcal{V}$ is denoted with the symbol $\mathcal{N}_v=\{w\in\V: P_{vw}>0\}$. 

Each agent in the network is associated with a pair $(x_v(t),y_v(t))\in \mathcal{X}\times \mathcal{Y}$ of variables, where $x_v(t) \in \mathcal{X}=[0,1]$ represents the individual \textit{hidden} inclination to a specific choice while $y_v(t) \in \mathcal{Y}=\{0,1\}$ is the effective realization of such choice at time $t\in\mathbb{N}$. Specifically, $x_v(t)\approx 0$ indicates opposition to making the target choice, and $x_v(t)\approx 1$ denotes support for it. Instead, $y_v(t)=1$ means that the agent makes the target choice at time $t\in\mathbb{Z}$ and $y_v(t)=0$ otherwise indicates that the agent has not made it.

We assume that each agent $v\in\mathcal{V}$ starts from an \textit{initial belief} $x_v(0) = x_0 $ that, at time step $ t\in\integernonnegative$, is updated according to two possible mechanisms. 

With probability $\alpha\in[0,1]$, the hidden inclination evolves as a response to social interactions and external factors, with the updated opinion becoming a convex combination of the opinions coming from the neighbors and the external influence. In particular, while the weight of the influence on $v \in \mathcal{V}$ from each neighbor \( w\in\mathcal{N}_v \) is given by \( P_{vw} \), the overall contribution of the neighbors is controlled by the parameter \( \lambda_{v}\in[0,1] \) and the updated opinion is given by
\begin{equation}
    x_{v}(t+1) = \lambda_{v} \sum_{w \in \mathcal{N}_{v}} P_{vw} x_{w}(t) + (1 - \lambda_{v}) \eta_{v}(t),
\end{equation}
where  \( \eta_{v}(t) \in [0,1]\) encodes the effect of external influences and individual biases at time \( t \). In particular, we decompose this last element as follows: 
\begin{equation}\label{eq:overall_external}
    \eta_{v}(t)={\eta^{0}_v+\eta_{v}^{\mathrm{nc}}(t)+u_{v}(t)},
\end{equation}
where $\eta^{\mathrm{0}}_v\in[0,1]$ is a deterministic constant indicating the \textit{individual bias} of the agent that, in the same spirit of Friedkin and Johnsen model \cite{friedkin1990social}, steers individual opinion along with the social influence in the absence of external influences. Instead, the uncontrollable input $\eta_{v}^{\mathrm{nc}}(t)$ represents possible (yet unpredictable) variations in the individual intrinsic bias (e.g., changes in the weather and, thus, the choice not to take a shared bike). In what follows, we make the following assumption on this uncontrollable component of our model.
\begin{assumption}\label{ass:noise}
    The uncontrollable inputs $\{\eta_{v}^{\mathrm{nc}}(t)\}_{v \in \mathcal{V}}$ are independent, identically distributed (i.i.d.) with zero mean and known covariance $\diag[\{\sigma_v^{2}\}_{v \in \mathcal{V}}]$ for all $t \in \integernonnegative$.
\end{assumption}
This assumption implies that changes induced by external factors cannot overcome the individual bias ${\eta}^{\mathrm{0}}_v$ but only slightly (and temporarily) modify it. Moreover, these changes are not correlated among different individuals, mimicking the fact that ultimately an external factor (e.g., the weather) can have a different impact on one's choices (e.g., renting a shared bike). Finally, the \textit{control input}
\begin{equation}\label{eq:stoc_constraint_input}
u_{v}(t) \in [0,1-\eta_{v}^0-\eta_{v}^{\mathrm{nc}}(t)],~v \in \mathcal{V},
\end{equation}
is what the policymaker can shape to nudge individual choices towards the target one, generally in a personalized fashion.
    
Instead, with probability $1-\alpha$, the agents' opinions are updated irrespective of social interaction, i.e., it becomes simply a convex combination of the previous individual opinion and external influences/initial bias as 
\begin{equation}\label{eq:dyn_B}
    x_{v}(t+1)=
    \lambda_{v}x_{w}(t)+(1-\lambda_{v})\eta_{v}(t),~v \in \mathcal{V},
\end{equation}
where $\lambda_{v} \in [0,1]$ now has the role of a relative weight that quantifies the effect of the external influence and initial biases on the evolution of the opinion. 

By defining the vector $\boldsymbol{x}(t) \in [0,1]^{|\mathcal{V}|}$ grouping the inclinations of all agents in the network and the vector $\boldsymbol{\eta}(t) \in [0,1]^{|\mathcal{V}|}$ of all external factors and biases, the previous dynamics can be written in matrix form as
\begin{subequations}\label{eq:overall_dyn}
\begin{equation}\label{eq:dyn}
    \boldsymbol{x}(t+1)=\boldsymbol{\Lambda}\boldsymbol{P}(t)\boldsymbol{x}(t)+(\boldsymbol{I}-\boldsymbol{\Lambda})\boldsymbol{\eta}(t),
\end{equation}
where 
\begin{equation}\label{eq:P_varying}
\boldsymbol{P}(t)=\begin{cases}
\boldsymbol{P}&\text{w.p. }\alpha\\
\boldsymbol{I}&\text{w.p. }1-\alpha
\end{cases}
\end{equation}
\end{subequations}
and $\boldsymbol{\Lambda}$ is a diagonal matrix with elements equal (and equally ordered) to parameters $\{\lambda_v\}_{v \in \mathcal{V}}$.

Let us further introduce the vectors $\boldsymbol{u}(t)$, $\boldsymbol{\eta}^0$ and $\boldsymbol{\eta}^{\mathrm{nc}}(t)$ grouping the control input, initial biases and the external factor acting on each agent, respectively. Throughout the remainder of the paper, we will assume that policymakers have a limited budget $\beta$ at their disposal to nudge agents' choices, i.e., the control input vector has to satisfy the following: 
\begin{equation}\label{eq:budget_constr}
    \boldsymbol{1}^{\top}\boldsymbol{u}(t)\leq \beta , \ \forall t\in\integernonnegative.
\end{equation}

\begin{remark}[Dynamics and external factors]
Assuming that the uncontrollable inputs are i.i.d. and independent of the state $\vx(t)$ is crucial to prove the properties of the dynamics in \eqref{eq:overall_dyn} analyzed in Section~\ref{sec:ergodicity}. Meanwhile, the assumption on the specific distribution and its covariance are mainly exploited in policy design.     
\end{remark}

Finally, we assume that the measured acceptance variables of the target choice, grouped in the vector $\boldsymbol{y}(t) \in \{0,1\}^{|\mathcal{V}|}$, are random variables as stated formally below.    
\begin{assumption}[Acceptance variables]\label{ass:Y}
Given
$\vx(t)\in[0, 1]$ for $t\in\mathbb{N}$, we assume the acceptance variables are random variables whose conditional distribution satisfies
\begin{subequations}\label{eq:acceptance_assump}
\begin{align}
&\mathbb{P}(y_v(t)=1|\boldsymbol{x}(t))=x_v(t),\\
&\mathbb{P}(y_v(t)=0|\boldsymbol{x}(t))=1-x_v(t),
\end{align}
\end{subequations}
for all $t\in\integernonnegative$.
\end{assumption}
Note that introducing these acceptance variables toward a target choice that a policy-maker intends to promote represents a situation in which only a noisy measure of individuals' opinions is available. In other words, the actual opinions of agents are not directly observable, and the available information is simply related to the manifested individual choices. 

\subsection{Novelty and relation to prior works}
The dynamics in \eqref{eq:dyn} is a novel extension of the classical Friedkin
and Johnsen model \cite{friedkin1990social} as well as its variations in opinion dynamics \cite{ravazzi2021ergodic}. The main distinguishing features of the model with respect to existing ones are: $(i)$ the presence of random factors, i.e., the random variable $\boldsymbol{\eta}^{\mathrm{nc}}(t)$, $t\in\integernonnegative$, modeling variability in individual inclinations due to external (uncontrollable) events; $(ii)$ the introduction of the dependence on a controlled input (through $\boldsymbol{u}(t)$, $t\in\integernonnegative$), mimicking the impact of policy actions to nudge a target choice; $(iii)$ the probabilistic update of opinions, according to the parameter $\alpha$ in \eqref{eq:P_varying}. 

Looking at how the controlled input impacts the evolution of opinions with respect to \cite{sprenger2024control}, in our model, the policymaker is not intended as a new node introduced in the network, able to influence individuals at each time with an equal action, but is meant as a (possibly personalized) external influence. In this sense, the control represents actions aimed at modifying the intrinsic bias of agent $v\in\mathcal{V}$, through, e.g., incentive strategies and information campaigns, rather than modifying their opinion directly. 

Meanwhile, considering \eqref{eq:P_varying} allows us to account for the fact that two processes, opinion updates related to social imitation and those related to individual decision making, can act on different time scales. In fact, social influence or pressure resulting from interactions that individuals have with their neighbors generally takes longer to have a significant effect, e.g., several days, weeks, or even months. Instead, individual decision-making may occur more frequently and on a shorter time scale, e.g., decisions may be made on a daily basis. The parameter $\alpha$ in \eqref{eq:P_varying} is used in our model to model the frequency of these two updating mechanisms. 

Lastly, our choice of introducing the acceptance variables (see Assumption~\ref{ass:Y}) allows us to distinguish our work from the literature related to robust and stochastic control \cite{fisher2009linear}. Indeed, instead of designing controllers taking into account that the opinion dynamics is uncertain and not perfectly observable, our approach is based on the framework of iterated random systems. In turn, this allows us to leverage a rich ergodic theory, that can be applied to study the system's long-run behavior and design nudging policies.

\subsection{Illustrative Scenarios}
Our model covers several aspects that customarily describe how people make repetitive decisions, e.g., daily choices of recycling garbage or using car/bike sharing services for daily commutations.

Indeed, one's choices are generally dictated by an intrinsic bias that is due to one's habits and/or socioeconomic status, which inherently makes one more or less inclined to repetitively make the same decision or adopt a certain service~(see the discussion in~\cite{Korteling2023}). The impact of such a bias can be mitigated by policies enacted by governments and stakeholders to favor new daily habits or facilitate access to a new service~\cite{Jaffe2022}, e.g., charging fees to people not recycling, or providing loyalty discounts for sharing services. In practice, these policies are nonetheless generally designed under budget constraints, capping investments that can be undertaken to nudge individuals to the target choice. This aspect is accounted for in our framework through the constraint in \eqref{eq:budget_constr}, in turn allowing us to consider more realistic limited budget scenarios. At the same time, uncontrollable events can temporarily shape one's inclination, e.g., weather can either hinder or stimulate the use of bike-sharing services, depending on whether it is bad or good. 

These external factors contribute to shaping individual opinions together with the social environment in which one is immersed. Indeed, in reality, acceptance of new habits or services is also driven by social dictates or homophily~\cite{McPherson2001}, with individual opinions being shaped by what is perceived as \textquotedblleft good\textquotedblright \ or \textquotedblleft bad\textquotedblright \ by the society as well as influenced by the benefits/downside of the targeted choice experiment by peers. Since it requires opinions on the targeted choice to form and shape as well as others make a certain choice or try services out, social influence is nonetheless likely to steer individual opinions on a different time scale with respect to the evolution of one inclination based on individual needs and irrespective of social dictates. 

Ultimately, when analyzing and designing policy strategies, policymakers and stakeholders generally devise or have access to surveys (see, e.g.,~\cite{JRC96151}) that allow them eventually to unveil the socioeconomic status and habits of an individual as well as the manifested acceptance/non-adoption of the targeted choice, but they generally do not give direct access to how one's opinion truly evolves. This is considered in our model through the introduction of the variable $\vy{(t)}$, $t \in \integernonnegative$, linked to the \textit{hidden} thanks to Assumption~\ref{ass:Y}.
  
\section{Asymptotic analysis}\label{sec:ergodicity} 
We first present a theoretical analysis of the behavior of the dynamics~\eqref{eq:dyn} in the absence of control, i.e., when $\boldsymbol{u}(t)=0$ for all $t\in\integernonnegative$. Despite this, we stress that two different sources of randomness have to be considered in this analysis. First, the social imitation mechanism acts in an intermittent way with a certain probability $\alpha$. Second, the variation in the individual bias is still characterized by $\boldsymbol{\eta}^{\mathrm{nc}}(t)$, $t\in \integernonnegative$, that captures the intrinsic volatility of the conditions at each time instant. These two components ultimately make both the latent and acceptance variables continue to oscillate over time. In addition, as $\vy(t)$ are random variables for all $t \in \integernonnegative$, the acceptance variables introduce another source of uncertainty. Nonetheless, under Assumption~\ref{ass:Y} and an additional hypothesis on the topology of the influence network, we will show that both the latent and acceptance variables converge to a final limit profile.

Specifically, we make the following assumption on the topology of the influence network associated with $\boldsymbol P$.
\begin{assumption}[Network topology]\label{ass:P}
For any node $v \in\mathcal{V}$, there exists a path from $v\in\mathcal{V}$ to a node $s\in\mathcal{V}$ such that $\lambda_s<1$. 
\end{assumption}
This assumption is common in opinion dynamics (see, e.g., \cite{ravazzi2021ergodic}), in principle allowing to have some agents that are not directly sensitive to external factors. At the same time, it also implies that each agent is influenced by at least another one, hence indirectly being affected by such external factors through repeated interactions with neighbors over time. This assumption is thus sufficient to guarantee that $\boldsymbol{\Lambda}\boldsymbol{P}(t)$ is, on average, Schur stable, allowing the expected opinions dynamics to converge as shown in the following proposition.
\begin{proposition}[Expected hidden inclinations' dynamics]\label{prop:convergence}
Let $\boldsymbol{u}(t)=0$ for all $z \in \integernonnegative$. Under Assumptions \ref{ass:noise} and \ref{ass:P}, the dynamics in \eqref{eq:overall_dyn} satisfy
\begin{equation}\label{eq:dyn_ev}
\mathbb{E}[\boldsymbol{x}(t+1)]=\boldsymbol{\Lambda}\overline{\boldsymbol{P}}\mathbb{E}[\boldsymbol{x}(t)]+(\boldsymbol{I}-\boldsymbol{\Lambda})\boldsymbol{\eta}^0,\end{equation} 
for any initial condition $ \boldsymbol{x}(0) \in\mathbb{R}^{|\mathcal{V}|}\in[0,1]$ and any $t\in\integernonnegative$. Moreover, $\lim_{t\rightarrow\infty}\mathbb{E}[\boldsymbol{x}(t)]=\vx^{\star}$ with
\begin{align}
\label{eq:equilibria}
\vx^{\star}&:=\mathbb{E}[{\boldsymbol x}(\infty)]=(\boldsymbol{I}-\boldsymbol\Lambda\overline{\boldsymbol P})^{-1}(\boldsymbol{I}-\boldsymbol\Lambda)\boldsymbol{\eta}^0\end{align}
and $\overline{\boldsymbol{P}}:=(1-\alpha)\boldsymbol{I}+\alpha\boldsymbol{P}$.
\end{proposition}
\begin{proof}By conditioning, from \eqref{ass:Y} we get
\begin{align*}
   \mathbb{E}[\boldsymbol{x}(t+1)]&=\mathbb{E}[\mathbb{E}[\boldsymbol{x}(t+1)|\vx(t)]]\\
   &=\mathbb{E}[\mathbb{E}[\boldsymbol{\Lambda}\boldsymbol{P}(t)\vx(t)|\vx(t)]+(\boldsymbol{I}-\boldsymbol{\Lambda})\veta^{0}\\&=\boldsymbol{\Lambda}\overline{\boldsymbol{P}}\mathbb{E}[\vx(t)]+(\boldsymbol{I}-\boldsymbol{\Lambda})\veta^{0},
\end{align*}
where the second equality stems from Assumption~\ref{ass:noise}. Meanwhile, based on Assumption~\ref{ass:P}, it follows that $\boldsymbol{\Lambda}\overline{\boldsymbol{P}}$ is Schur stable, with $\overline{\boldsymbol{P}}:=(1-\alpha)\boldsymbol{I}+\alpha\boldsymbol{P}$. Consequently, we get 
$\lim_{t\rightarrow\infty}\mathbb{E}[\boldsymbol{x}(t+1)]=\vx^{\star}$, thus concluding the proof.
\end{proof}

By relying on this result, as well as on Assumption~\ref{ass:Y}, we can then formalize the following result on the acceptance variables.
\begin{corollary}[Expected acceptance variables]\label{prop:convergence2}
Let $\vu(t)=0$ for all $t \in \integernonnegative$. Under Assumption~\ref{ass:noise}-~\ref{ass:P} it holds that
$\mathbb{E}[\boldsymbol{y}(t+1)]=\mathbb{E}[\boldsymbol{x}(t+1)]$
and  
\begin{equation}
\lim_{t\rightarrow\infty}\mathbb{E}[\boldsymbol{y}(t)]=\vx^{\star},
\end{equation}
for any initial condition $ \boldsymbol{x}(0) \in\mathbb{R}^{|\mathcal{V}|}\in[0,1]$.
\end{corollary}
\begin{proof}
The proof follows straightforwardly from Assumption~\ref{ass:Y} and the result in Proposition~\ref{prop:convergence} and, thus, it is omitted.
\end{proof}
These two results imply that opinions in a social system evolving according to \eqref{eq:dyn} are stable in expectation in the absence of externally controlled action. Moreover, they show that, in the absence of controlled inputs, a final opinion limit profile emerges that is a combination of the expected value of exogenous (uncontrollable) inputs, i.e., $\veta_0(t)$ for $t \in \integernonnegative$, and the expected impact of social interactions. 

After introducing these properties, we can now summarize the features of the dynamics in \eqref{eq:overall_dyn} as follows.
\begin{theorem}[Ergodicity of hidden inclinations]
\label{thm:ergodic}
Consider the random process $\{\boldsymbol x(t)\}_{t\in\integernonnegative}$ defined in \eqref{eq:overall_dyn} and let $\vu(t)=0$, $\forall t \in \integernonnegative$. Under Assumptions \ref{ass:noise} and ~\ref{ass:P}, the following hold.
\begin{enumerate}
\item $\boldsymbol x(k)$ converges in distribution to random variables $\boldsymbol x_{\infty}$ and the distribution of $\boldsymbol x_{\infty}$, is the unique invariant
distribution for \eqref{eq:overall_dyn}.
\item The process is ergodic.
\item The limit random variables satisfy
$
\Exp[\boldsymbol x_\infty]=\boldsymbol x^{\star}.
$ 
\end{enumerate}
\end{theorem}
\begin{proof}
Let us rewrite \eqref{eq:dyn} as
\begin{align}\label{eq:dyn2}
    \boldsymbol{x}(t+1)&=\boldsymbol{A}(t)\boldsymbol{x}(t)+\boldsymbol{b}(t),
\end{align}
with $\boldsymbol{A}(t)=\boldsymbol{\Lambda}\boldsymbol{P}(t)$ and $\boldsymbol{b}(t)=(\boldsymbol{I}-\boldsymbol{\Lambda})\boldsymbol{\eta}(t)$. As we assume that the controlled input is zero, then it holds that
$\mathbb{E}[\boldsymbol{b}(t)]=(\boldsymbol{I}-\boldsymbol{\Lambda})\boldsymbol{\eta}^0
$, which is constant. From Assumption~\ref{ass:P}, we further obtain that $\mathbb{E}[\boldsymbol{A}(t)]=\boldsymbol{\Lambda}\overline{\boldsymbol{P}}$ is Shur stable. By direct application of \cite[Theorem~1]{ravazzi2015ergodic} we deduce that
$\boldsymbol x(k)$ converges in distribution to a random variable $\boldsymbol x_{\infty}$, and the distribution of $\boldsymbol x_{\infty}$ is the unique invariant distribution for \eqref{eq:dyn2}. The process \eqref{eq:dyn2} is ergodic and the limit random variable satisfies
$
\Exp[\boldsymbol x_\infty]=\boldsymbol x^{\star}.
$
\end{proof} 
The following theorem guarantees that the ergodicity of the latent opinions implies the ergodicity of the acceptance variables. 
\begin{theorem}[Ergodicity of acceptance variables]\label{thm:y}
Consider the random processes $\{\boldsymbol x(t)\}_{t\in\integernonnegative}$ defined in~\eqref{eq:dyn} and $\{\boldsymbol y(t)\}_{t\in\integernonnegative}$. Suppose that $\vu(t)=0$ for all $t \in \integernonnegative$. Then, under Assumptions~\ref{ass:noise}-\ref{ass:P}, the following
hold.
\begin{enumerate}
\item $\{\boldsymbol y(t)\}_{t\in\integernonnegative}$ converges in distribution to random variables $\boldsymbol y_{\infty}$ that is distributed as a Bernoulli distribution with parameter $\boldsymbol x_{\infty}$. %, whose distribution is the unique invariant distribution for \eqref{eq:dyn2}. 
In particular, $\Exp[\boldsymbol y_\infty]=\boldsymbol x^{\star}$.
\item The process $\{\boldsymbol y(t)\}_{t\in\integernonnegative}$ is ergodic.
\end{enumerate}
\end{theorem}
\begin{proof}
It is well known that \(X_n \xrightarrow{\mathcal{D}} X\) 
(convergence in distribution) if and only if for every bounded and continuous function 
\(h: \mathbb{R} \to \mathbb{R}\), we have:
$
\mathbb{E}[h(X_n)] \to \mathbb{E}[h(X)].$
By exploiting Assumption~\ref{ass:Y}, for every bounded and continuous function $h(\cdot)$, it holds that:
\begin{align*}
\mathbb{E}[h(y_v{(t)})] &= \mathbb{E} \left[ \mathbb{E} [h(y_v(t)) | x_v{(t)}] \right]\\
 &= \mathbb{E} \left[h(1)x_v{(t)}+h(0)(1-x_v{(t)}) \right]\\
 &= h(1)\mathbb{E} \left[x_v{(t)}\right]+h(0)\mathbb{E} \left[(1-x_v{(t)}) \right].
\end{align*}
Since $ \vx{(t)}$ converges to $\vx_{\infty}$ in distribution (see Theorem~\ref{thm:ergodic}), for the properties of $h(\cdot)$ we have
\begin{align*}
\lim_{t \to \infty} \mathbb{E}[h(y_v(t))] &= h(1)  \mathbb{E}[\vx_{v}(\infty)] + h(0)  \mathbb{E}[1 - \vx_{v}(\infty)]\\
&=\mathbb{E}[h(\vy_{v}(\infty))],
\end{align*}
in turn, implying that
$
\lim_{t \to \infty} \mathbb{E}[h(\vy{(t)})] = \mathbb{E}[h(\vy_{\infty})].
$ 
From this fact we conclude that $\{\vy(t)\}_{t \in \integernonnegative}$ converges in distribution to $\vy_{\infty}$,  that is distributed as a Bernoulli distribution with parameter $\boldsymbol x_{\infty}$.

To complete the proof, it remains to show
that the process \( \{\boldsymbol y(t)\}_{t \in \integernonnegative} \), where \( \boldsymbol y(t) \sim \text{Bernoulli}(\boldsymbol x(t)) \), is also ergodic, i.e.,
\begin{align*}
\lim_{t\rightarrow\infty}\frac{1}{t} \sum_{\ell = 0}^{t - 1} \vy(\ell) = \mathbb{E}[\boldsymbol y_{\infty}] =\vx^{\star} \quad \text{almost surely}.
\end{align*}

By the Chernoff bound \cite{hoeffding1963}, for any \(\epsilon > 0\) the following holds

\[
\mathbb{P}\left( \left\| \frac{1}{t} \sum_{\ell = 0}^{t-1} \boldsymbol y(\ell) - \frac{1}{t} \sum_{\ell = 0}^{t-1} \mathbb{E}[\boldsymbol y(\ell)] \right\|_1 \geq \epsilon \right) \leq 2 \exp \left( - \frac{\epsilon^2 t}{2|\mathcal{V}|} \right).
\]
If we choose $\epsilon=t^{-1/2+\alpha}$ for any $\alpha>0$, then \[
\sum_{t=1}^{\infty} 2 \exp \left( - \frac{\epsilon_t^2 t}{2|\mathcal{V}|} \right)= \sum_{t=1}^{\infty} 2 \exp \left( - \frac{t^{2\alpha} }{2|\mathcal{V}|} \right)< \infty
\]
the
Borel-Cantelli Lemma (see \cite{Borkar1995}, Theorem 1.4.2) implies that with
probability one $\left\| \frac{1}{t} \sum_{\ell = 0}^{t-1} \boldsymbol y(\ell) - \frac{1}{t} \sum_{\ell = 0}^{t-1} \mathbb{E}[\boldsymbol y(\ell)] \right\|_1$ for all but finitely many values of $t\in\mathbb{N}$. Therefore, almost surely $\left\| \frac{1}{t} \sum_{\ell = 0}^{t-1} \boldsymbol y(\ell) - \frac{1}{t} \sum_{\ell = 0}^{t-1} \mathbb{E}[\boldsymbol y(\ell)] \right\|_1$ converges to zero as $t$ goes to infinity.
Finally, from triangular inequality and since \(\mathbb{E}[\boldsymbol y(\ell)]\!=\!\mathbb{E}[\boldsymbol x(\ell)]\) for all $\ell\in\mathbb{N}$,  we have
\begin{align*}
 \left\| \frac{1}{t} \sum_{\ell = 0}^{t-1} \boldsymbol y(\ell) - \vx^{\star} \right\|_1  &\leq
\left\| \frac{1}{t} \sum_{\ell = 0}^{t-1} \boldsymbol y(\ell) - \frac{1}{t} \sum_{\ell = 0}^{t-1} \mathbb{E}[\boldsymbol y(\ell)] \right\|_1\\
&+\left\| \frac{1}{t} \sum_{\ell = 0}^{t-1} \mathbb{E}[ \boldsymbol{x}(\ell)]-\boldsymbol{x}^{\star} \right\|_1.    
\end{align*}

Since the first term in the right-hand-side converges to zero almost surely as $t\to\infty$ from arguments above and \(\{\boldsymbol x(t)\}_{t \in \integernonnegative}\) is ergodic (see Theorem~\ref{thm:ergodic}), we have that
\[
\frac{1}{t} \sum_{\ell = 0}^{t-1} \boldsymbol y(\ell) \to \mathbb{E}[\boldsymbol x_{\infty}] \!=\! \mathbb{E}[\boldsymbol y_{\infty}] =\vx^{\star}\  \text{almost surely as } t \to \infty.
\]
and hence \(\{\boldsymbol y(t)\}_{t \in \integernonnegative}\) is ergodic, thus concluding the proof.
\end{proof}

\begin{remark}[Open-loop properties and control actions]\label{remark:constant_policy}
    The properties summarized in 
    Theorems \ref{thm:ergodic}-\ref{thm:y} are still valid if one feeds the social system with a constant control input over time, i.e., $\boldsymbol{u}(t)=\boldsymbol{u}$ for all $t\in\integernonnegative$.  
\end{remark}

Based on this consideration, in the next section, we focus on the design of constant nudging policies, which preserve the ergodicity of both the hidden and manifested inclinations and are then used as a benchmark in Section~\ref{sec:numerical} against an (optimal) time-varying one. 

\section{Asymptotic behavior with constant control action: model-free versus model-based approach}\label{sec:CCP}
We now introduce two possible constant policies that the policymaker/stakeholder can enact to nudge a target choice. The first one is a \textit{model-free}, constant control policy (MFCCP), obtained by simply dividing the overall budget $\beta$ (see \eqref{eq:budget_constr}) in a uniform way. The second is an \textit{optimal} policy designed by accounting for the asymptotic behavior of the dynamics in \eqref{eq:overall_dyn} and the fact that only the acceptance variables (see Assumption~\ref{ass:Y}) are accessible. In both cases, one would require the (generally unrealistic) knowledge of the stochastic disturbance $\boldsymbol{\eta}^{\mathrm{nc}}(t)$ to account for it by design while satisfying $\boldsymbol{x}(t) \in [0,1]$ at all $t \in \integernonnegative$. Based on Assumption~\ref{ass:noise}, we here take a (simplistic and conservative) approach to do so by imposing the designed policy to lay in the following interval
\begin{subequations}\label{eq:worst_case_bound}
\begin{equation}\label{eq:uniform_policy}
\boldsymbol{u}(t) \in [\boldsymbol{0},\boldsymbol{1}-\boldsymbol{\eta}'],   
\end{equation}
where
\begin{equation}
\eta_{v}'=\eta_{v}^{0}+2\sigma_{v},~~\forall v \in \mathcal{V},
\end{equation}
\end{subequations}
and the second term depends on the standard deviation of the external disturbance $\boldsymbol{\eta}^{\mathrm{nc}}(t)$ for all $t \in \integernonnegative$.
\begin{remark}[On the usefulness of \eqref{eq:worst_case_bound}]
    The bound in \eqref{eq:worst_case_bound} leverages the implicit assumption that $\boldsymbol{\eta}^{\mathrm{nc}}(t)$ has a light-tailed distribution, e.g., Gaussian, making it unpractical to cope with disturbances with heavy-tailed distributions, e.g., Cauchy distributed disturbances. 
\end{remark}

The most straightforward policy that could be deployed thus corresponds to a uniform distribution of the available resources (see \eqref{eq:budget_constr}), eventually saturated based on the individual worst-case bound in \eqref{eq:worst_case_bound}. This policy can be formalized as follows and we will refer to it as Model Free Constant Control Policy (MFCCP).
\begin{policy}[MFCCP]\label{MFCCP}
The MFCCP is given by 
\begin{equation}\label{eq:MFCCP}
    u_v(t)\!=\!u^{\mathsf{MF\!}}_v\!=\!\min{\left\{1-\eta_{v}',\frac{\beta}{|\V|}\right\}},~~\forall v \in \V,~\forall t\in\integernonnegative.
\end{equation}\end{policy}
However, due to its nature, a constant policy like MFCCP is \textquotedblleft blind\textquotedblright \ with respect to the imitation and influence effect due to interactions in the social network, which can instead be leveraged to nudge the target choice achieving the compromise generally sought by policymakers and stakeholders, i.e., maximizing the acceptance of a specific choice (i.e., maximizing the policy's \textit{efficiency}) while minimizing the costs of a policy. Indeed, disregarding social interactions under budget constraints may cause excessive (eventually unneeded) resource investments in some individuals that can be more \textquotedblleft easily\textquotedblright \ influenced by others, while sufficient efforts are not invested in nudging other ones that instead would need more effort on the policymaker/stakeholder side to be convinced.    

To account for social interactions in pursuing the compromise between the two opposite objectives of maximal efficiency and minimal cost, we introduce a personalized constant control policy, , which we will refer to as Model Based Constant Control Policy (MBCCP). The latter is designed by leveraging the result in Theorem~\ref{thm:y} and Remark~\ref{remark:constant_policy} to optimally achieve a trade-off between the asymptotic average acceptance of the target choice and policy effort, as follows.   
\begin{policy}[MBCCP]\label{prob:loss_inf}
The MBCCP is obtained as
\begin{subequations}\label{eq:MBCCP}
\begin{equation}
\begin{aligned}
\vu(t)=\boldsymbol{u}^{\mathsf{MB}}&=\argmin{\boldsymbol{u}}{J_{\mathsf{MB}}(\boldsymbol{u})}\\
&\qquad\quad~ \mathrm{s.t.}~ \boldsymbol{1}^{\!\top}\boldsymbol{u}\leq\beta,~\boldsymbol{u} \in [\boldsymbol{0},\boldsymbol{1}-\boldsymbol{\eta}'],
\end{aligned}
\end{equation}
for all $t \in \integernonnegative$, with
\begin{equation}\label{eq:loss_inf}
J_{\mathsf{MB}}(\boldsymbol{u}):=\mathbb{E}\left[\|\boldsymbol{1}-{\vy}_{\infty}\|_{\boldsymbol{Q}}^{2}\!+\!\|{\boldsymbol{u}}\|_{\boldsymbol{R}}^{2}\right],
\end{equation}
\end{subequations}
and $\boldsymbol{Q},\boldsymbol{R}$ being user-defined positive definite matrices.
\end{policy} 
The policy design objective in \eqref{eq:loss_inf} can further be proven to be solely a function of the policy $\vu$ to be designed and the initial individual bias $\boldsymbol{\eta_{0}}$ as formalized in the following result. 
\begin{proposition}\label{prop:CQP}
Under Assumption~\ref{ass:noise}-\ref{ass:P}, the control objective in \eqref{eq:loss_inf} can be equivalently recast as
\begin{subequations}\label{eq:reformulated_loss}
    \begin{equation}
        J_{\mathsf{MB}}({\vu})= \|\vu\|_{\boldsymbol{W}}^{2} +\boldsymbol{c}^{\top}\vu+\gamma,
    \end{equation}
    where
    \begin{align}
        & \boldsymbol{W}=\boldsymbol{R}+\boldsymbol{
        V}^{\top\!}(\boldsymbol{Q}-[\boldsymbol{Q}])\boldsymbol{
        V},\\
        & \boldsymbol{c}^{\top}=\boldsymbol{1}^{\!\top\!}([\boldsymbol{Q}]-2\boldsymbol{Q})\boldsymbol{V}+2\veta_{0}^{\top}(\boldsymbol{Q}-[\boldsymbol{Q}])\boldsymbol{V}\\
        & \gamma=\boldsymbol{1}^{\!\top\!}\boldsymbol{Q}\boldsymbol{1}\!+\!\veta_{0}^{\top}\boldsymbol{
        V}^{\top\!}(\boldsymbol{Q}\!-\![\boldsymbol{Q}])\boldsymbol{
        V}\veta_{0}\!+\!\boldsymbol{1}^{\!\top\!}([\boldsymbol{Q}]\!-\!2\boldsymbol{Q})\boldsymbol{V}\veta_{0},
    \end{align}
    with $\boldsymbol{V}=(I-\boldsymbol{\Lambda\overline{P}})^{-1}(\boldsymbol{I}-\boldsymbol{\Lambda})$.
\end{subequations}
\end{proposition}
\begin{proof}
Let us rewrite the first term in \eqref{eq:loss_inf} as follows:
    \begin{align*}
    \mathbb{E}[\|\boldsymbol{1}-\boldsymbol{y}_{\infty}\|_{\boldsymbol{Q}}^{2}]&=\mathbb{E}[\mathbb{E}[\|\boldsymbol{1}-\boldsymbol{y}_{\infty}\|_{\boldsymbol{Q}}^{2}\big|\boldsymbol{x}_{\infty}]\\
    &=\|\boldsymbol{1}-\vx^\star\|_{\boldsymbol{Q}}^{2}+\langle\vx^{\star},\boldsymbol{1}-\vx^{\star}\rangle_{[\boldsymbol{Q}]}
    \end{align*}
    where the second equality follows from Theorem~\ref{thm:y}. Due to the presence of the constant input $\vu$ to be optimized, similarly to \eqref{eq:equilibria}, $\vx^{\star}$ now satisfies:
    \begin{align*}
    \vx^{\star}&=\boldsymbol{\Lambda}\overline{\boldsymbol{P}}\vx^{\star}+(\boldsymbol{I}-\boldsymbol{\Lambda})(\veta_{0}+\vu)\\
        & \Rightarrow \vx^{\star}=(\boldsymbol{I}-\boldsymbol{\Lambda}\overline{\boldsymbol{P}})^{-1}(\boldsymbol{I}-\boldsymbol{\Lambda})(\veta_{0}+\vu)=\boldsymbol{V}(\veta_{0}+\vu).
    \end{align*}
Putting the expression $\vx^{\star}$ into \eqref{eq:loss_inf} and grouping quadratic terms in $\vu$, linear terms in $\vu$ and those that are instead independent of the quantity to be optimized, the result straightforwardly follows.
\end{proof}
This reformulation highlights that the MBCCP is the solution to a quadratic program (QP) with linear inequality constraints, that can be efficiently tackled with several well-known approaches, including interior point methods,
active set, augmented Lagrangian, and conjugate gradient methods \cite{boyd2004convex}, to mention just a few.

Based on the reformulation of the loss provided in Proposition~\ref{prop:CQP}, it would also possible to find explicit expressions for $\boldsymbol{u}_{\mathsf{MB}}$ via the Karush-Kuhn-Tucker (KKT) conditions associated with \eqref{eq:MBCCP}, i.e.,
\begin{subequations}
    \begin{align}
        &2 \boldsymbol{W}\boldsymbol{u}^{\mathsf{MB}}+c+\nu_{1}\boldsymbol{1}-\boldsymbol{\nu_{2}}+\boldsymbol{\nu_{3}}=0,\label{eq:conditionMB1}\\
        & \nu_{1} \!\geq\! 0,~~\nu_{2,v} \!\geq\! 0,~~\nu_{3,v} \!\geq\! 0,~~~\forall v \in \{1,\ldots,|\mathcal{V}|\},\label{eq:conditionMB2}\\
        & \nu_{1}(\boldsymbol{1}^{\!\top\!}\boldsymbol{u}^{\mathsf{MB}}\!-\!\beta)\!-\!\!\sum_{v=1}^{|\mathcal{V}|}\nu_{2,v}u_{v}^{\mathsf{MB}}\!+\!\!\sum_{v=1}^{|\mathcal{V}|}\nu_{2,v}(u_{v}^{\mathsf{MB}}\!-\!1+\eta_{v}')
        =0,\label{eq:conditionMB3}\\
        & \boldsymbol{1}^{\top}\boldsymbol{u}^{\mathsf{MB}}-\beta \leq 0,~~ -\boldsymbol{u}^{\mathsf{MB}}\leq \mathbf{0},~~ \boldsymbol{u}^{\mathsf{MB}}\leq \boldsymbol{1}-\boldsymbol{\eta}',\label{eq:conditionMB4}
\end{align}
\end{subequations}
where $\nu_{1}$ is the Lagrange multiplier associated with the budget constraint, while $\boldsymbol{\nu_{2}}$ and $\boldsymbol{\nu_{3}}$ stack the multipliers associated with the lower and upper bounds on each component of $\boldsymbol{u}^{\mathsf{MB}}$ dictated by \eqref{eq:worst_case_bound}, respectively. Accordingly, there are several possible scenarios for the (constant) value ultimately taken by $\boldsymbol{u}_{\mathsf{MB}}$, depending on the combination of active constraints induced by the user-defined matrices $\boldsymbol{Q}$ and $\boldsymbol{R}$ (see \eqref{eq:loss_inf}), as well as the available budget $\beta$. 

\begin{remark}[MBCCP with inactive constraints]
Whenever the constraints are not active (i.e., the conditions in \eqref{eq:conditionMB4} are strict inequalities), the Lagrange multipliers have to be zero for \eqref{eq:conditionMB2} and \eqref{eq:conditionMB3} to be satisfied. This result ultimately leads to $\boldsymbol{u}^{\mathsf{MB}}=-\frac{1}{2}\boldsymbol{W}^{-1}\boldsymbol{c}$. Given a fixed budget $\beta$, the previous result implies the following condition
\begin{equation}
    -\frac{1}{2}\begin{bmatrix}
        \boldsymbol{1}^{\top}\\
        -\boldsymbol{I}\\
        \boldsymbol{I}
    \end{bmatrix}\boldsymbol{W}(\boldsymbol{Q},\boldsymbol{R})^{-1}\boldsymbol{c}(\boldsymbol{Q})<\begin{bmatrix}
        \beta\\
        \boldsymbol{0}\\
        \boldsymbol{1}-\boldsymbol{\eta}'
    \end{bmatrix},
\end{equation}
where we have explicitly highlighted the dependence of the weights $\boldsymbol{W}$ and $\boldsymbol{c}$ on the user-defined ones, i.e., $\boldsymbol{Q}$ and $\boldsymbol{R}$ (see \eqref{eq:loss_inf}). Such a condition can hence be checked by the policymaker/stakeholder to understand whether the model-based constant policy they design allows them $(i)$ not to consume the overall budget, $(ii)$ to nudge everyone, and $(iii)$ to be careful with respect to external disturbances.   
\end{remark}

As already pointed out in Remark~\ref{remark:constant_policy}, whether the policymaker decides to adopt the MFCCP in \eqref{eq:MFCCP} or the MBCCP, the stability in expectation of the opinion dynamics, ergodicity of both the hidden inclination and the manifest variables, as well as the convergence in the mean square sense of the dynamics straightforwardly follow with minor adjustments (i.e., replacing $\veta_0$ with $\veta_0+\vu$) from Theorems~\ref{thm:ergodic}-\ref{thm:y}.
\section{Adaptive policy design: a model predictive approach}\label{sec:control}
While the MBCCP strategy given by \eqref{prob:loss_inf} already seeks an optimal trade-off between the efficiency and costs of a policy, it does so by restrictively looking at achieving an efficiency-cost balance in the long run and a steady-state scenario, neglecting the intertwined evolution of opinions and policies over time. We now make a step forward in overcoming this limitation, formulating the following policy design problem explicitly optimizing the efficiency-cost trade-off accounting for how the acceptance of the target choice evolves over time under the influence of the nudging policy.
\begin{policy}[Infinite-horizon policy with quadratic cost]\label{policy:infinite}
    Let the policymaker/stakeholder aiming to maximize the average realizations of unitary acceptance variables $\vy(t)$ while minimizing the policy effort $\boldsymbol{u}(t)$, for all $t \in\integernonnegative$. A possible optimal policy to achieve this objective can be designed by solving
    \begin{equation}\label{eq:infinite_hor_opt}
    \begin{aligned} &\min_{\boldsymbol{U}_{\infty}} ~J(\boldsymbol{U}_{\infty})\\
    &~~ \mathrm{s.t.~} \eqref{eq:overall_dyn} \mathrm{~and~}\eqref{eq:acceptance_assump},~~\forall t \in \integernonnegative,\\
    &\qquad~ \boldsymbol{u}(t) \in [\boldsymbol{0},\boldsymbol{1}-\boldsymbol{\eta}'],~~\forall t \in \integernonnegative,\\ 
    & \qquad~  \boldsymbol{1}^{\top}\boldsymbol{u}(t)\leq \beta, ~~ \forall t \in \integernonnegative,
    \end{aligned}
\end{equation}
where $\boldsymbol{U}_{\infty}\!=\!\{\boldsymbol{u}(t)\}_{t\in\integernonnegative}$ comprises all the policy actions to be designed over time, and the infinite-horizon loss $J(\boldsymbol{U}_{\infty})$ is given by:
\begin{equation}\label{eq:infinite_cost1}
J(\boldsymbol{U}_{\infty})=\sum_{t=0}^{+\infty} \left[\mathbb{E}\|\boldsymbol{1}-{\vy}(t)\|_{\boldsymbol{Q}}^{2}+\|{\boldsymbol{u}}(t)\|_{\boldsymbol{R}}^{2}\right],
\end{equation}
where $\boldsymbol{Q}$ and $\boldsymbol{R}$ are positive definite matrices shaping the trade-off between efficiency and control effort. 
\end{policy}
Note that this infinite policy formulation looks at maximizing manifested acceptance, i.e., the only information that can be accessed by the policymaker. Nonetheless, under Assumption~\ref{ass:Y}, this quantity is still closely related to hidden inclinations variables $\vx(t)$. 
\begin{remark}[Relation with MBCCP]
    The MBCCP (see Policy~\ref{prob:loss_inf}) corresponds to an infinite horizon policy with quadratic cost (similar to Policy~\ref{policy:infinite}) where the transient in the objective function is neglected. 
\end{remark}

As done for the MBCCP in Section~\ref{sec:CCP}, we first focus on the cost $J(\boldsymbol{U}_{\infty})$ and analyze its properties. To this end, let us introduce the following compact notation for the expected value of the hidden inclination variables: 
\begin{equation}\label{eq:average}
    \boldsymbol{\mu}(t)=\mathbb{E}[\vx(t)].
\end{equation}
By exploiting this formalism, we can introduce the formalize the following result
\begin{proposition}[Equivalent cost]\label{prop:loss1}
 Under Assumption~\ref{ass:Y}, the cost function in \eqref{eq:infinite_cost1} is equivalent to
 \begin{equation}\label{eq:loss_in_mu}
    J_{\infty\!}(\boldsymbol{U}_{\infty})\!=\!\!\!\sum_{t=0}^{+\infty}\!\left[\|\boldsymbol{1}\!\!-\!\!\boldsymbol{\mu}(t)\|_{\boldsymbol{Q}\!}^{2}+\!\|\boldsymbol{u}(t)\|_{\boldsymbol{R}}^{2}+\!\langle\boldsymbol{\mu}(t),\!\boldsymbol{1}\!\!-\!\!\boldsymbol{\mu}(t)\rangle_{[\boldsymbol{Q}]}\right]\!.
 \end{equation}
\end{proposition}
\begin{proof}
    Let us focus on the first term of the cost in \eqref{eq:infinite_cost1} only. The latter can be recast as
    \begin{equation*}
        \mathbb{E}[\|\boldsymbol{1}-\boldsymbol{y}(t)\|_{\boldsymbol{Q}}^{2}]= \mathbb{E}[\mathbb{E}[\|\boldsymbol{1}-\boldsymbol{y}(t)\|_{\boldsymbol{Q}}^{2}\big|\boldsymbol{x}(t)]].
    \end{equation*}
    Thanks to Assumption~\ref{ass:Y}, standard computations allow us to rewrite this term further as a function of $\boldsymbol{\mu}(t)$ in \eqref{eq:average} as follows:
    \begin{align*}
        &\mathbb{E}[\mathbb{E}[\|\boldsymbol{1}-\boldsymbol{y}(t)\|_{\boldsymbol{Q}}^{2}\big|\boldsymbol{x}(t)]]=\|\boldsymbol{1}-\boldsymbol{\mu}(t)\|_{\boldsymbol{Q}}^{2}+\\
        &\quad -\sum_{v \in \V}Q_{vv}(1-\mu_v(t))^{2}\!+\!\!\sum_{v \in \V}Q_{vv}\mathbb{E}[\mathbb{E}[(1-y_v(t))^2]\big|\boldsymbol{x}(t)]\\
        &\quad=\|\boldsymbol{1}-\boldsymbol{\mu}(t)\|_{\boldsymbol{Q}}^{2}-\sum_{v \in \V}\!Q_{vv}(1-\mu_v(t))^{2}\\
        &\quad+\sum_{v \in \V}Q_{vv}\mathbb{E}[1-\mu_v(t)]\\
        &\quad=\|\boldsymbol{1}-\boldsymbol{\mu}(t)\|_{\boldsymbol{Q}}^{2}+\langle\boldsymbol{\mu}(t),\boldsymbol{1}-\boldsymbol{\mu}(t)\rangle_{[\boldsymbol{Q}]}.
    \end{align*}
     By replacing with this last expression the first term in \eqref{eq:infinite_cost1}, we straightforwardly get \eqref{eq:loss_in_mu} and conclude the proof.
\end{proof}
As shown in Proposition~\ref{prop:loss1}, the cost in \eqref{eq:infinite_cost1} ultimately consists of three terms. The first is a quadratic function in the expected values of the agents' inclinations, the second is a quadratic form in the control input, and the third term allows us to equivalently recast the expected quadratic cost in a quadratic loss of expected values. 

The reformulation of the cost in \eqref{eq:loss_in_mu} further enables us to confine the infinite horizon cost in \eqref{eq:infinite_cost1} within an interval, as summarized below.
\begin{lemma}[Bound on the infinite-horizon cost]\label{lemma:bound_cost}
    For any $\boldsymbol{U}_{\infty}$ satisfying the constraints in \eqref{eq:infinite_hor_opt}, the loss in \eqref{eq:infinite_cost1} satisfies
\begin{subequations}
     \begin{equation}
         \underline{J}_{\infty}(\boldsymbol{U}_{\infty})\leq J_{\infty}(\boldsymbol{U}_{\infty})\leq \bar{J}_{\infty}(\boldsymbol{U}_{\infty}),
     \end{equation}  where
     \begin{align}
         & \underline{J}_{\infty}(\boldsymbol{U}_{\infty})= \sum_{t=0}^{+\infty}\!\left[\|\boldsymbol{1}\!-\!\boldsymbol{\mu}(t)\|_{\boldsymbol{Q}\!}^{2}+\!\|\boldsymbol{u}(t)\|_{\boldsymbol{R}}^{2}-\|\boldsymbol{1}\|_{[\boldsymbol{Q}]}^{2}\right]\!,\label{eq:lower_bound_cost}\\
         & \bar{J}_{\infty}(\boldsymbol{U}_{\infty})=\sum_{t=0}^{+\infty}\!\left[\|\boldsymbol{1}\!-\!\boldsymbol{\mu}(t)\|_{\boldsymbol{Q}\!}^{2}+\!\|\boldsymbol{u}(t)\|_{\boldsymbol{R}}^{2}+\|\boldsymbol{1}\|_{[\boldsymbol{Q}]}^{2}\right]\!.\label{eq:upper_bound_cost}
     \end{align}
\end{subequations}
\end{lemma}
\begin{proof}
   Let us focus on the last term in \eqref{eq:loss_in_mu}, verifying
   \begin{equation*}
    \langle\boldsymbol{\mu}(t),\boldsymbol{1}\!-\!\boldsymbol{\mu}(t)\rangle_{[\boldsymbol{Q}]}=(\boldsymbol{\mu}(t))'[\boldsymbol{Q}]\boldsymbol{1}-\|\boldsymbol{\mu}(t)\|_{[\boldsymbol{Q}]}^{2},
   \end{equation*}
   by definition, with $[\boldsymbol{Q}]$ having all non-negative entries since $\boldsymbol{Q}$ is positive definite. Since $\boldsymbol{U}_{\infty}$ satisfies the constraints imposed in \eqref{eq:infinite_hor_opt}, then $\boldsymbol{\mu}(t)\in [0,1]^{|\V|}$ for all $t \in \integernonnegative$. These two results imply that 
   \begin{equation}\label{eq:useful_bound}
      0 \leq (\boldsymbol{\mu}(t))'[\mathbf{Q}]\boldsymbol{1}\leq \|\boldsymbol{1}\|_{[\mathbf{Q}]}^{2},~~\forall t \in \integernonnegative, 
   \end{equation} 
    and thus, the following upper-bound holds:
   \begin{equation*}
    \langle\boldsymbol{\mu}(t),\boldsymbol{1}\!-\!\boldsymbol{\mu}(t)\rangle_{[\boldsymbol{Q}]} \leq (\boldsymbol{\mu}(t))'[\boldsymbol{Q}]\boldsymbol{1} \leq \|\boldsymbol{1}\|_{[\boldsymbol{Q}]}^{2},~~ \forall t \in \integer.
   \end{equation*} 
   The bound in \eqref{eq:upper_bound_cost} straightforwardly follows. Meanwhile, the third term in \eqref{eq:loss_in_mu} also satisfies 
   \begin{equation*}
    \langle \boldsymbol{\mu}(t),\boldsymbol{1}\!-\!\boldsymbol{\mu}(t)\rangle_{[\boldsymbol{Q}]} \geq -\|\boldsymbol{\mu}\|_{[\boldsymbol{Q}]}^{2} \geq  -\|\boldsymbol{1}\|_{[\boldsymbol{Q}]}^{2},~~ \forall t \in \mathbb{Z}_{\geq 0}, 
   \end{equation*}
    where the first inequality stems from the lower-bound in \eqref{eq:useful_bound} while the last one holds thanks to the bounds on $\boldsymbol{\mu}(t)$ guaranteed by the assumption on $\boldsymbol{U}_{\infty}$. The bound in \eqref{eq:lower_bound_cost} easily follows, ending the proof.  
\end{proof}
This result shows that the cost function in \eqref{eq:loss_in_mu} can be both lower and upper bounded by a quadratic cost in $\{\boldsymbol{\mu}(t)\}_{t \in \integernonnegative}$ and $\boldsymbol{U}_{\infty}$ except for a constant, that changes only in sign depending on whether we consider the lower or the upper bound. 

While ideally guaranteeing an optimal efficiency-cost trade-off (according to \eqref{eq:infinite_cost1}) over an infinite horizon, Policy~\ref{policy:infinite} is knowingly intractable from a computational perspective even for small networks, because it entails the solution of an optimization problem with infinite constraints. Moreover, it might not truly fit the policymakers' and stakeholders' needs. Indeed, practical policies are likely never enacted for an infinite period of time, but rather redesigned at times based on societal and technological changes (e.g., substantial technological advancements). Therefore, we now shift from infinite-horizon to finite-horizon policy design. 

\begin{remark}[Shift to finite-horizon policies]\label{remark:shift}
    If $\boldsymbol{Q}=[\boldsymbol{Q}]$, Proposition~\ref{prop:loss1} implies that the cost to be optimized is quadratic in $\boldsymbol{U}_{\infty}$ but only linear in $\{\boldsymbol{\mu}(t)\}_{t \in \integernonnegative}$. To avoid this different treatment of the controlled input and expected inclinations, we approximate the cost by considering only the quadratic terms in $\{\boldsymbol{\mu}(t)\}_{t \in \integernonnegative}$ and $\boldsymbol{U}_{\infty}$ common to \eqref{eq:lower_bound_cost} and \eqref{eq:upper_bound_cost} in going from an infinite-horizon to a finite-horizon policy.
\end{remark}

\subsection{Oracle finite-horizon control policy}
We now step away from the impractical, infinite-horizon Policy~\ref{policy:infinite} to a practical finite-horizon one, that still accounts that the individual inclinations are time-varying and exhibit an oscillating behavior. To this end, we initially assume that an \textit{oracle} provides the policymaker/stakeholder with the expected values $\boldsymbol{\mu}(t)$ of the hidden inclinations at all time instant $t \in \integernonnegative$. Under this assumption, by exploiting the result in Proposition~\ref{prop:convergence}, we cast the following finite-horizon policy design problem.

\begin{policy}[Oracle finite-horizon policy]\label{Policy_MPC}
Given a finite horizon $T \in \integernonnegative$, with $T>1$, the input $\vu(t)$ to be fed to the social system to nudge individuals toward a target choice at time $t \in \integernonnegative$ can be designed by solving the following constrained QP:
\begin{equation}\label{eq:conservative_MPC}
    \begin{aligned}
        & \min_{\boldsymbol{U_{|t}}} J_{T}(\boldsymbol{U_{|t}})\\
        &~~\mathrm{s.t.~}
        \boldsymbol{\mu}(k\!+\!\!1|t)\!=\!\boldsymbol{\Lambda}\overline{\boldsymbol{P}}\boldsymbol{\mu}(k|t)\!+\!(\boldsymbol{I}\!-\!\!\boldsymbol{\Lambda})(\boldsymbol{\eta}_{0}\!+\!\boldsymbol{u}(k|t)),~\forall k \!\in\! \mathcal{I}_{T},\\
        & \qquad~ \boldsymbol{u}(k|t) \in [\boldsymbol{0},\boldsymbol{1}-\boldsymbol{\eta}'],~~~\forall k \in \mathcal{I}_{T},\\
        & \qquad~ \boldsymbol{1}^{\!\top}\boldsymbol{u}(k|t)\leq \beta,~~\forall k \in \mathcal{I}_{T}],\\
        & \qquad~ \boldsymbol{\mu}(0|t)=\boldsymbol{\mu}(t),\\
        &\qquad~ \boldsymbol{\mu}(T|t)=(\boldsymbol{I}-\boldsymbol{\Lambda}\overline{\boldsymbol{P}})^{-1}(\boldsymbol{I}-\boldsymbol{\Lambda})(\boldsymbol{\eta}_0+\boldsymbol{u}^{\mathsf{MB}}),
    \end{aligned}
\end{equation}
where $\mathcal{I}_{T}=[0,T-1]$, $\boldsymbol{u}^{\mathsf{MB}}$ is the solution to \eqref{eq:MBCCP}, $\boldsymbol{\mu}(k|t)$ is the vector of expected inclinations predicted according to the first set of constraints starting from the initial condition dictated by $\boldsymbol{\mu}(t)$, and $\boldsymbol{U_{|t}}=\{\boldsymbol{u}(k|t)\}_{k=0}^{T-1}$ is the set of inputs to be optimized. The cost $J_{T}(\boldsymbol{U_{|t}})$ to be minimized is instead 
\begin{equation}\label{eq:cost_1}
    J_{T}(\boldsymbol{U_{|t}})=\sum_{k=0}^{T-1} \left[\|\boldsymbol{1}-\boldsymbol{\mu}(k|t)\|_{\boldsymbol{Q}}^{2}+\|\boldsymbol{u}(k|t)\|_{\boldsymbol{R}}^{2}\right],
\end{equation}
according to Remark~\ref{remark:shift}.
\end{policy}
In principle, solving \eqref{eq:conservative_MPC} results in a sequence of $T$ inputs $\boldsymbol{U_{|t}}^{\star}$ that can be used to nudge individual over the user-defined finite horizon of length $T$. Nonetheless, enacting the policy in a recording horizon fashion, i.e., by setting $\vu(t)=\vu^{\star}(0|t)$ and discarding the rest of $\boldsymbol{U_{|t}}^{\star}$, is advisable given the external uncontrollable factors temporarily influencing individual choices and our simplistic approach to account for them. 

For the receding horizon oracle policy, we can prove the following formal result.
\begin{theorem}[Oracle policy's properties]
Let the initial expected inclination given by the oracle be $\boldsymbol{\mu}(0) \in [0,1]^{\mathcal{V}}$ and assume that the optimal control problem in \eqref{eq:conservative_MPC} is feasible at time $t =0$. Then, the following properties hold for all positive definite $\boldsymbol{Q}$ and $\boldsymbol{R}$.
\begin{itemize}
    \item[$i)$] The optimal control problem in \eqref{eq:conservative_MPC} is recursively feasible $\forall t \geq 0$.
    \item[$ii)$] The expected inclinations asymptotically converge to the steady-state \(\boldsymbol{\mu}^{\mathsf{MB}}=(\boldsymbol{I} - \boldsymbol{\Lambda}\overline{\boldsymbol{P}})^{-1}(\boldsymbol{I} - \boldsymbol{\Lambda})(\boldsymbol{\eta}_0 + \boldsymbol{u}^{\mathsf{MB}}) \), with $\boldsymbol{u}^{\mathsf{MB}}$ being the solution of \eqref{eq:MBCCP}.
\end{itemize}
\end{theorem}
\begin{proof}
We first prove recursive feasibility (i.e., $i)$). To this end, assume that \eqref{eq:conservative_MPC} is feasible at $t \in \integernonnegative$ and denote the associated optimal control sequence and the predicted expected inclinations as
\begin{align*}
&\boldsymbol{U_{|t}}^\star = \{\boldsymbol{u}^\star(0|t), \boldsymbol{u}^\star(1|t), \ldots, \boldsymbol{u}^\star(T-1|t)\}, \\
&\boldsymbol{M_{|t}}^\star = \left\{\boldsymbol{\mu}^\star(0|t), \boldsymbol{\mu}^\star(1|t), \ldots, \boldsymbol{\mu}^\star(T-1|t),{\boldsymbol{\mu}}^{\mathsf{MB}} \right\}, 
 \end{align*}
respectively, where the last term in $\boldsymbol{M_{|t}}^\star$ is induced by the terminal constraint. Let us construct the candidate solution for the instance of \eqref{eq:conservative_MPC} at time $t+1$ as
\begin{equation}\label{eq:candidate_input}
    \boldsymbol{U_{|t+1}}' = \left\{\boldsymbol{u}^\star(1|t), \boldsymbol{u}^\star(2|t), \ldots, \boldsymbol{u}^\star(T-1|t),{\boldsymbol{u}}^{\mathsf{MB}}\right\},
\end{equation}
where ${\boldsymbol{u}}^{\mathsf{MB}}$ is the solution of \eqref{eq:MBCCP}, satisfying both the box constraint ${\boldsymbol{u}}^{\mathsf{MB}} \in [0,\boldsymbol{1}-\veta']$ and the budget constraint by construction. As the other elements of $\boldsymbol{U_{|t+1}}'$ correspond to the last $T-1$ elements of $\boldsymbol{U_{|t}}^\star$, also satisfying both the budget constraint and the box constraints, the candidate input sequence is feasible. Let us now define the expected inclination sequence associated with $\boldsymbol{U_{|t+1}}'$, i.e.,    
\begin{equation}\label{eq:candidate_state}
    \boldsymbol{M_{|t+1}}' = \left\{\boldsymbol{\mu}^\star(1|t), \boldsymbol{\mu}^\star(2|t), \ldots, {\boldsymbol{\mu}}^{\mathsf{MB}},{\boldsymbol{\mu}}'(T|t+1)\right\},
\end{equation}
where the first $T-1$ elements are a feasible trajectory by construction (as they are the last $T$ elements of $\boldsymbol{M_{|t}}^{\star}$). Meanwhile, since the ${\boldsymbol{u}}'(T-1|t+1)={\boldsymbol{u}}^{\mathsf{MB}}$ by construction and ${\boldsymbol{\mu}}'(T-1|t+1)={\boldsymbol{\mu}}^{\mathsf{MB}}$, then $\boldsymbol{\mu}'(T|t+1)={\boldsymbol{\mu}}^{\mathsf{MB}}$ thanks to the stationarity of the pair $({\boldsymbol{\mu}}^{\mathsf{MB}},{\boldsymbol{u}}^{\mathsf{MB}})$. Therefore, $\boldsymbol{M_{|t+1}}'$ is a feasible expected inclination sequence. Since this candidate can be defined for all $t \in \integernonnegative$, a feasible solution for \eqref{eq:conservative_MPC} exists $\forall t \in \integernonnegative$ and property $i)$ easily follows.

We now focus on property $ii)$. To prove it, let us equivalently rewrite the cost in \eqref{eq:conservative_MPC} as follows
\begin{equation*}
    J_{T}(\boldsymbol{U_{|t}})=\sum_{k=0}^{T-1} \ell(\boldsymbol{\mu}(k|t),{\boldsymbol{\mu}}^{\mathsf{MB}},\boldsymbol{u}(k|t),{\boldsymbol{u}}^{\mathsf{MB}}),
\end{equation*}
where
\begin{align*}
      &\ell(\boldsymbol{\mu}(k|t),{\boldsymbol{\mu}}^{\mathsf{MB}},\boldsymbol{u}(k|t),{\boldsymbol{u}}^{\mathsf{MB}})\!=\!\|\boldsymbol{\varepsilon_{\mu}}(k|t)\!-\!\boldsymbol{\varepsilon_{\mu}}^{\mathsf{MB}}\|_{\boldsymbol{Q}}^{2}\!+\!\|\boldsymbol{\varepsilon_{u}}\|_{\boldsymbol{R}}^{2}\\
    &\qquad \quad+2(\boldsymbol{\varepsilon_{\mu}}^{\mathsf{MB}})^{\top}\boldsymbol{Q}(\boldsymbol{\varepsilon_{\mu}}(k|t)\!-\!\boldsymbol{\varepsilon_{\mu}}^{\mathsf{MB}})+2(\boldsymbol{u}^{\mathsf{MB}})^{\top}\boldsymbol{R}\boldsymbol{\varepsilon_{u}}(k|t)\\
    &\qquad \quad +\|\boldsymbol{\varepsilon_{\mu}}^{\mathsf{MB}}\|_{\boldsymbol{Q}}^{2}\!+\!\|\boldsymbol{u}^{\mathsf{MB}}\|_{\boldsymbol{R}}^{2}\geq 0,
\end{align*}
with $\boldsymbol{\varepsilon_{\mu}}(k|t)\!=\!\boldsymbol{1}\!-\boldsymbol{\mu}(k|t)$, $\boldsymbol{\varepsilon_{\mu}}^{\mathsf{MB}}\!=\!\boldsymbol{1}\!-\!{\boldsymbol{\mu}}^{\mathsf{MB}}$ and $\boldsymbol{\varepsilon_{u}}(k|t)\!=\!\boldsymbol{u}(k|t)\!-\!{\boldsymbol{u}}^{\mathsf{MB}}$. Note that, $\ell(\boldsymbol{\mu}(k|t),{\boldsymbol{\mu}}^{\mathsf{MB}},\boldsymbol{u}(k|t),{\boldsymbol{u}}^{\mathsf{MB}})$ satisfies the following
\begin{align}\label{eq:useful_inequality}
\nonumber \ell(\boldsymbol{\mu}(k|t),{\boldsymbol{\mu}}^{\mathsf{MB}},\boldsymbol{u}(k|t),{\boldsymbol{u}}^{\mathsf{MB}})&\geq \ell({\boldsymbol{\mu}}^{\mathsf{MB}},{\boldsymbol{\mu}}^{\mathsf{MB}},{\boldsymbol{u}}^{\mathsf{MB}},{\boldsymbol{u}}^{\mathsf{MB}})\\
&=\|\boldsymbol{\varepsilon_{\mu}}^{\mathsf{MB}}\|_{\boldsymbol{Q}}^{2}\!+\!\|\boldsymbol{u}^{\mathsf{MB}}\|_{\boldsymbol{R}}^{2},
\end{align}
for all $k \in \mathcal{I}_{T}$, since the first two terms are always non-negative and become zero (i.e., assume their minimal value) only when $\boldsymbol{\varepsilon_{\mu}}(k|t)=\boldsymbol{\varepsilon_{\mu}}^{\mathsf{MB}}$. Moreover, let us compactly denote the cost associated with the optimal solution at a given time instant $t$ and the one associated with the candidate solution as $J_{T|t}^{\star}=J_{T}(\boldsymbol{U_{|t}}^{\star})$ and $J_{T|t}'=J_{T}(\boldsymbol{U_{|t}}')$, respectively. Note that, by optimality, it is straightforward that $J_{T|t}^{\star} \leq J_{T|t}'$ at all $t \in \integernonnegative$.

By considering the candidates in \eqref{eq:candidate_input} and \eqref{eq:candidate_state}, note that the following holds:
\begin{align*}
    J_{T|t+1\!}'&\!=\!J_{T|t}^{\star\!}\!-\!\ell(\boldsymbol{\mu}(0|t),\!{\boldsymbol{\mu}}^{\mathsf{MB}\!},\boldsymbol{u}(0|t),\!{\boldsymbol{u}}^{\mathsf{MB}})\!+\!\|\boldsymbol{\varepsilon_{\mu}}^{\mathsf{MB}}\|_{\boldsymbol{Q}}^{2}\!+\!\|\boldsymbol{u}^{\mathsf{MB}}\|_{\boldsymbol{R}}^{2}\\
    &\leq J_{T|t}^{\star}
\end{align*}
where the inequality follows from \eqref{eq:useful_inequality}. Accordingly, 
\begin{equation*}
    J_{T|t+1}^{\star} \leq J_{T|t+1}' \leq J_{T|t}^{\star},
\end{equation*}
and hence 
\begin{equation*}
    J_{T|t+1}^{\star} - J_{T|t}^{\star} \leq 0,
\end{equation*}
i.e., the optimal cost of \eqref{eq:conservative_MPC} is non-increasing. As a consequence, there exists a non-negative scalar $J_{T|\infty}$ such that $\lim_{t \rightarrow \infty} J_{T|t}=J_{T|\infty}$. Therefore,  
\begin{align*}
    J_{T|t+1}^{\star} - J_{T|t}^{\star} \geq &\ell(\boldsymbol{\mu}(0|t),{\boldsymbol{\mu}}^{\mathsf{MB}},\boldsymbol{u}(0|t),{\boldsymbol{u}}^{\mathsf{MB}})\\
    &\qquad \quad -\|\boldsymbol{\varepsilon_{\mu}}^{\mathsf{MB}}\|_{\boldsymbol{Q}}^{2}-\|\boldsymbol{u}^{\mathsf{MB}}\|_{\boldsymbol{R}}^{2}\geq 0,
\end{align*}
where the first term tends to zero for $t \rightarrow \infty$, implying that $\lim_{t \rightarrow \infty} \boldsymbol{\mu}(t)={\boldsymbol{\mu}}^{\mathsf{MB}}$ and $\lim_{t \rightarrow \infty} \boldsymbol{u}(t)={\boldsymbol{u}}^{\mathsf{MB}}$, thus concluding the proof.
\end{proof}
\begin{remark}
As well-known~(see, e.g.,\cite{Bemporad2000}), the explicit solution to \eqref{eq:conservative_MPC} is a piecewise affine (PWA) control law on the expected inclination, with the values of its gains depending on the active constraints at each time instant. Nonetheless, based on the KKT conditions for QP, one can also determine (if they are well-defined) the conditions on the weights $\boldsymbol{Q}$ and $\boldsymbol{R}$ guaranteeing the constraints in \eqref{eq:conservative_MPC} are always inactive (i.e., the bound characterizing them are never hit). In this case, the control law becomes a static feedback on the average inclinations. In this simpler scenario, our future research will be devoted to a formal analysis of the closed-loop properties of the resulting controlled system.
\end{remark}

\subsection{Policy design with estimated mean inclinations}
All previous policy design approaches rely on the exact knowledge of the average inclination $\boldsymbol{\mu}(t)$, $\forall t \in \integernonnegative$, which is generally not directly accessible to the policymaker. Indeed (see Assumption~\ref{ass:Y}), policymakers and stakeholders can generally access the acceptance variables $\vy(t)$, for $t \in \integernonnegative$. The latter can be used to define a practical estimate of the mean inclination over time as 
\begin{equation}\label{eq:estimate}
    \hat{\boldsymbol{\mu}}(t)=\frac{1}{t}\sum_{\tau=0}^{t-1} \vy(\tau).
\end{equation}
Despite the simplicity of this estimator, if the ergodicity of the system is preserved under the designed control action, the Ces\'aro time-average \eqref{eq:estimate} converges to their expected value, namely
\begin{equation}
    \mathbb{E} [\|\hat{\boldsymbol{\mu}}(t)-\boldsymbol{\mu}(t)\|_2^2]\rightarrow0.
\end{equation}
We can therefore substitute the oracle predictive control problem in \eqref{eq:conservative_MPC} with the following:
\begin{equation}\label{eq:conservative_MPC2}
    \begin{aligned}
        & \min_{\boldsymbol{U_{|t}}} J_{T}(\boldsymbol{U_{|t}})\\
        &~~\mathrm{s.t.~}
        \boldsymbol{\mu}(k\!+\!\!1|t)\!=\!\boldsymbol{\Lambda}\overline{\boldsymbol{P}}\boldsymbol{\mu}(k|t)\!+\!(\boldsymbol{I}\!-\!\!\boldsymbol{\Lambda})(\boldsymbol{\eta}_{0}\!+\!\boldsymbol{u}(k|t)),~\forall k \!\in\! \mathcal{I}_{T},\\
        & \qquad~ \boldsymbol{u}(k|t) \in [\boldsymbol{0},\boldsymbol{1}-\boldsymbol{\eta}'],~~~\forall k \in \mathcal{I}_{T},\\
        & \qquad~ \boldsymbol{1}^{\!\top}\boldsymbol{u}(k|t)\leq \beta,~~\forall k \in \mathcal{I}_{T},\\
        & \qquad~ \boldsymbol{\mu}(0|t)=\hat{\boldsymbol{\mu}}(t),\\
        &\qquad~ \boldsymbol{\mu}(T|t)=(\boldsymbol{I}-\boldsymbol{\Lambda}\overline{\boldsymbol{P}})^{-1}(\boldsymbol{I}-\boldsymbol{\Lambda})(\boldsymbol{\eta}_0+\boldsymbol{u}^{\mathsf{MB}}),
    \end{aligned}
\end{equation}
i.e., replace the \textquotedblleft true\textquotedblright \ yet unknown average inclinations with their estimates, computed as in \eqref{eq:estimate}.
\section{Numerical Results}\label{sec:numerical}
\begin{figure}[!tb]
  \centering
\includegraphics[scale=.4,trim=4cm 3cm 3cm 2cm,clip]{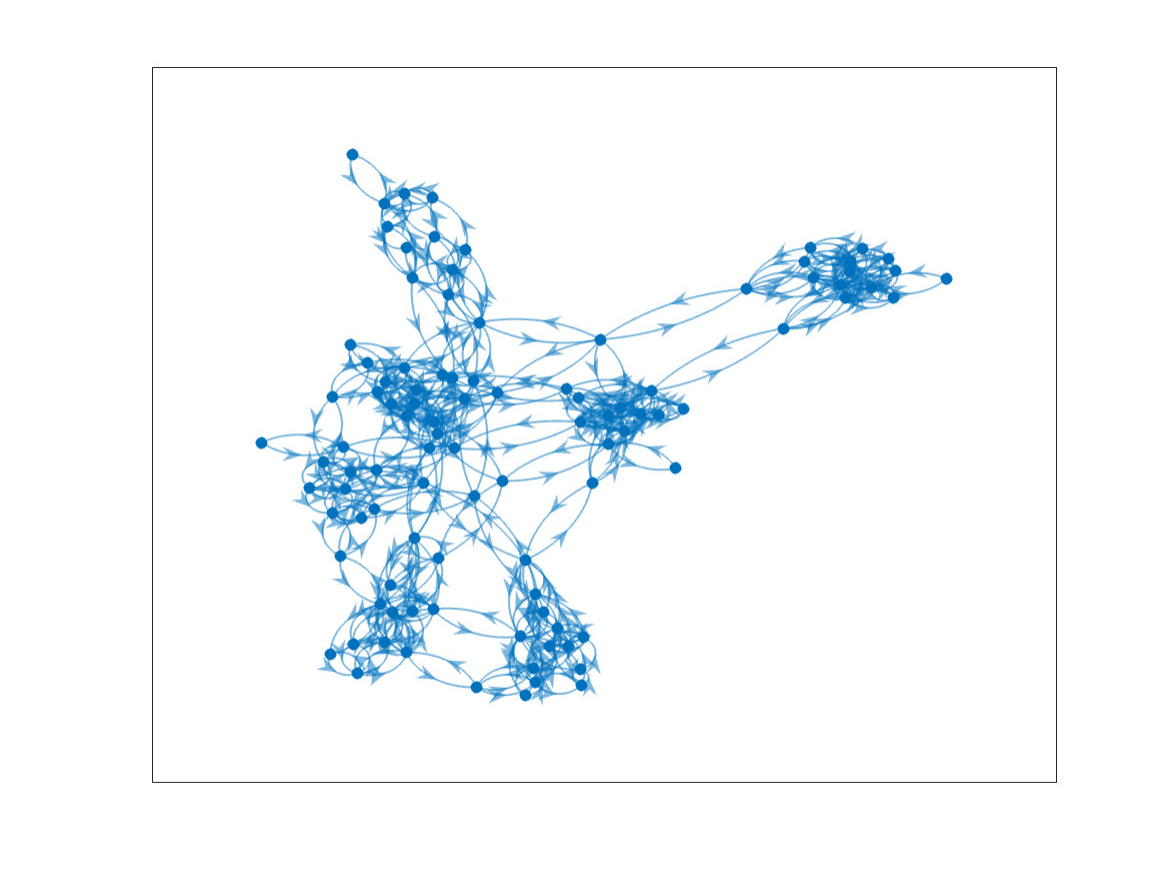} 
  \caption{The social network with $100$ nodes considered in the numerical example. }
  \label{fig:example}
\end{figure}
To illustrate the free evolution of the model proposed in Section~\ref{sec:model}, as well as the effect of the control policies proposed in Sections~\ref{sec:CCP} and \ref{sec:control}, let us consider the randomly generated\footnote{We set the density of the edges $|\mathcal{E}|/|\mathcal{V}|^2$ to 0.05 and the probability of connections between nodes in different clusters to $\gamma=0.9$.} modular graph consisting of $|\mathcal{V}|=100$ nodes,  grouped in 7 clusters depicted in \figurename{~\ref{fig:example}}, representing our social network. This set of nodes is split into two groups. One (smaller) group $\mathcal{V}_+\subset \mathcal{V}$ of $10$ nodes is favorably inclined toward a target choice of interest, having $x_v(0)=0.7$, for all $v\in\mathcal{V}_+$. The remaining 90\% of nodes are instead negatively inclined with respect to the target choice, with each of these 90 nodes having an initial condition sampled at random from a uniform distribution over the interval $[0,0.1]$. Note that this implies that the considered population is mainly negatively biased with respect to the target choice.

We test the evolution of inclinations within this population according to the model proposed in Section~\ref{sec:model} in open and closed loop, in this second scenario adopting (and comparing) policies designed according to the strategies proposed in Sections~\ref{sec:CCP}-\ref{sec:control}. In both cases, we represent external and uncontrollable factors affecting individual inclinations as realizations of a zero-mean Gaussian distributed sequence with a standard deviation of 0.1.

With respect to the parameters in \eqref{eq:overall_dyn}, the entries of the diagonal matrix $\boldsymbol{\Lambda}$ are drawn from a uniform distribution over the interval $[0, 1]$. The time-scale separation parameter $\alpha$ is instead chosen differently depending on whether we consider the free or the controlled evolution of \eqref{eq:overall_dyn}. While in the first case, if fixed to a (low) value, in the second one, it is varied to assess the impact of different time-scale separation parameters when closing the loop with the policies introduced in Sections~\ref{sec:CCP}-\ref{sec:control}.

\begin{figure}[!tb]
\centering
\begin{tabular}{cc}
\subfigure[Hidden inclinations \textit{vs} time]{\includegraphics[width=0.475\columnwidth]{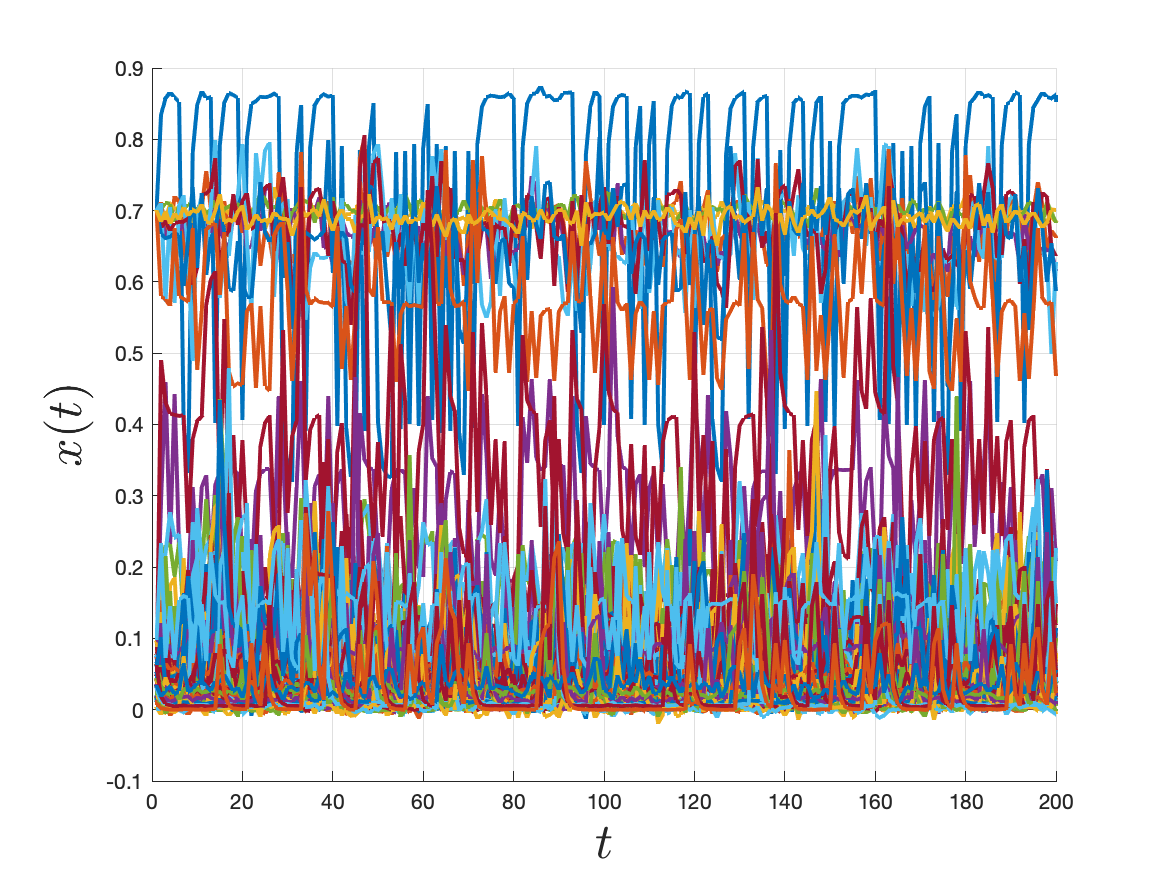}}& \subfigure[Manifested opinions \textit{vs} time]{\includegraphics[width=0.475\columnwidth]{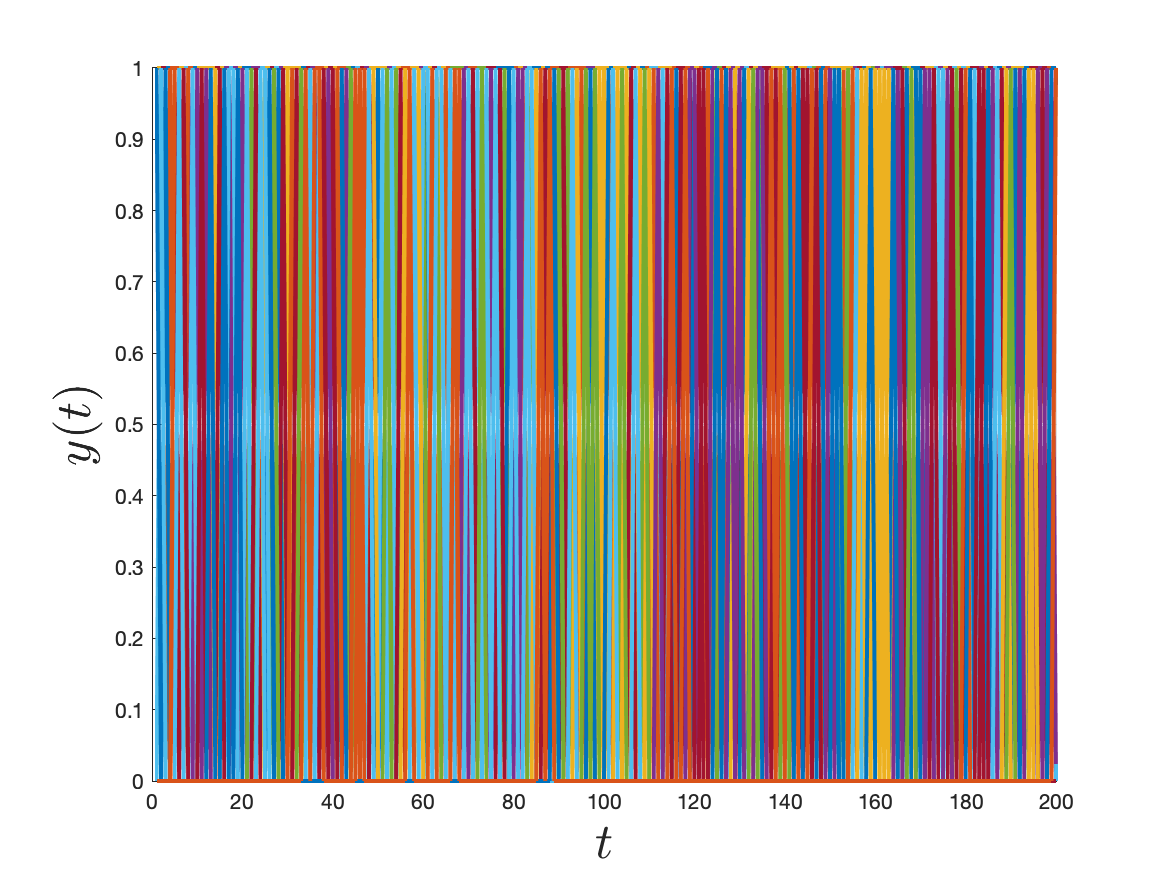}}\\
\subfigure[Time average of hidden inclinations \textit{vs} time]{\includegraphics[width=0.475\columnwidth]{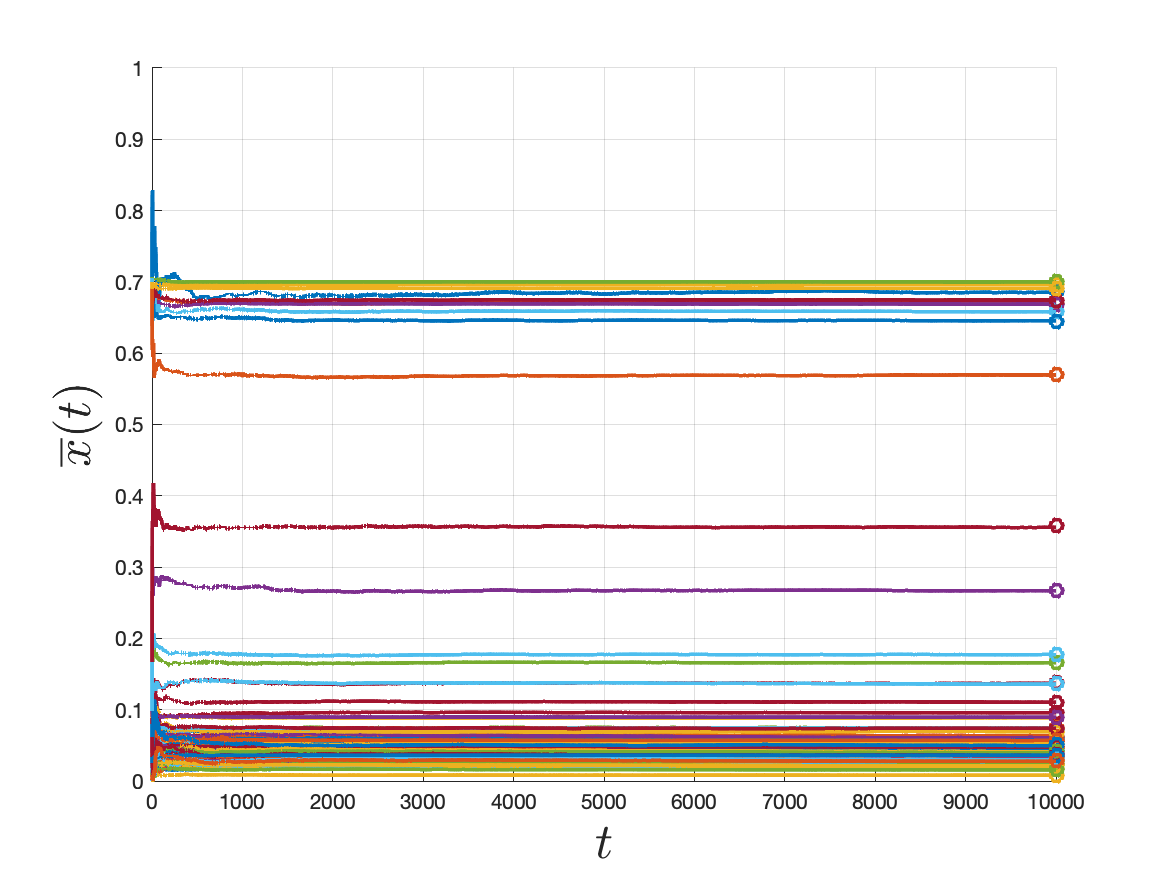}}
&
\subfigure[Time average of manifested opinions \textit{vs} time]{\includegraphics[width=0.475\columnwidth]{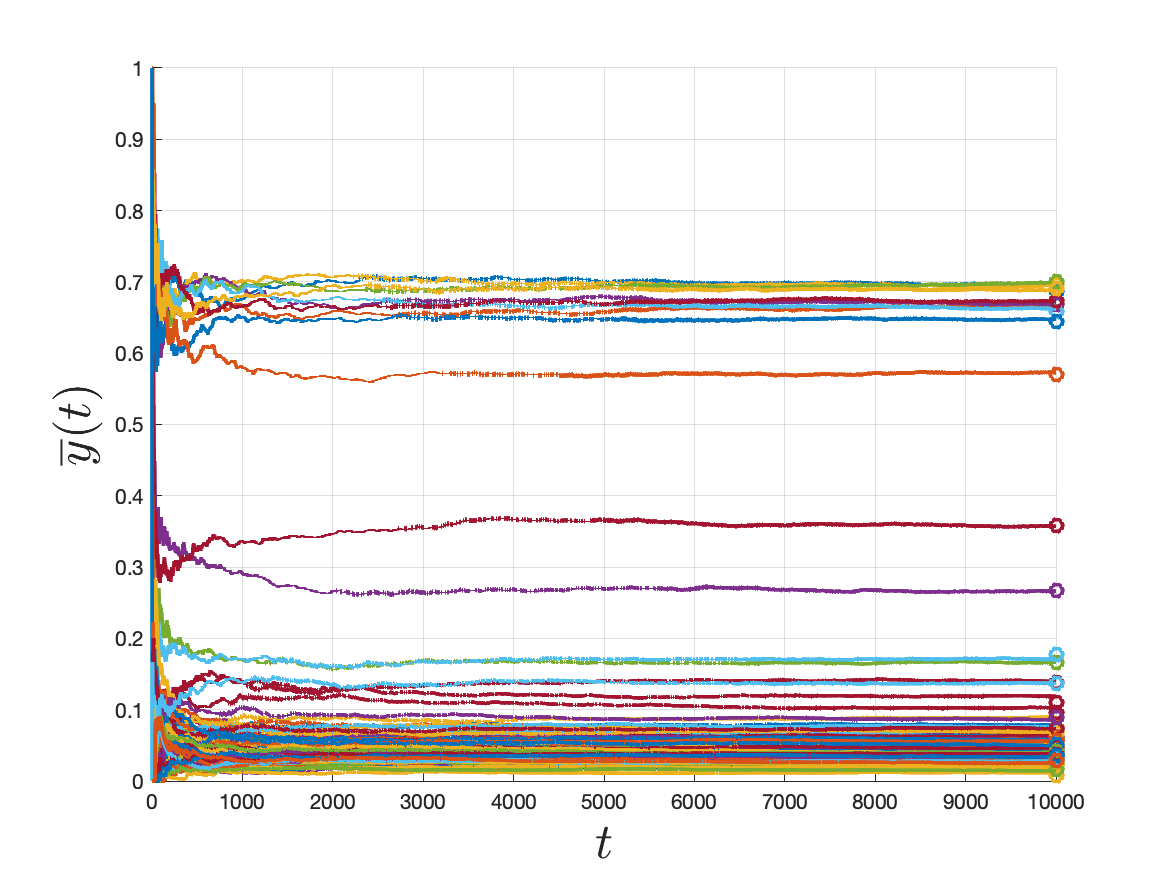}}
\end{tabular}
\caption{Free evolution of the system: evolution of hidden and manifested options and their time averages over time.}
\label{fig:free}
\end{figure}
\paragraph{Free evolution of the social system} \figurename{~\ref{fig:free} } showcases the free evolution of the hidden and manifest opinions, respectively characterized in \eqref{eq:overall_dyn} and Assumption~\ref{ass:Y}, when $\alpha$ is set to 0.25. These results confirm the theoretical results formalized in Theorem~\ref{thm:ergodic} and Theorem~\ref{thm:y}. Indeed, as the latent opinions, i.e., $\vx(t)$ and the acceptance variables $\vy(t)$ continue to oscillate over time (with $t \in \integernonnegative$), the corresponding time averages $\overline{\vx}(t)$ and $\overline{\vy}(t)$ see \eqref{eq:time_average}) converge to the same expected opinion profile.
\subsection{Performance analysis of control policies}
We now close the loop and compare the impact of non-personalized policies (i.e., MFCCP in \eqref{eq:MFCCP}), personalized policies neglecting transients (see \eqref{eq:MBCCP}), and the moving-horizon strategy presented in Section~\ref{sec:control} against the free evolution of the social system. In our tests, we consider a simulation horizon of $T_{\mathrm{sim}}=20$ steps, impose a budget $\beta=10$, and (as previously mentioned) vary the time-scale separation parameter $\alpha\in\{0.25,0.5,0.75,1\}$. In designing both the MBCCP policy in \eqref{eq:MBCCP} and the moving-horizon strategy with estimated average inclination (see~\eqref{eq:conservative_MPC2}), we impose 
\begin{equation}\label{eq:test_weights}
 \boldsymbol{Q}=\boldsymbol{I},~~~ \boldsymbol{R}=r\boldsymbol{I}.
\end{equation}
In our closed-loop simulations $r$ is varied within the interval $(0,5]$ to consider a spectrum of strategies with different penalization of policy costs. When employing the policy in \eqref{eq:conservative_MPC2}, we set the prediction horizon $T$ to $5$.     
\paragraph{Key performance indicators} To evaluate the cost-benefit trade-off of the different policies (against the free evolution of the system), we consider two indexes. 

The first indicator is defined as follows
\begin{equation}\label{eq:social_cost}
\Gamma_{T_{\mathrm{sim}}}=\sum_{t=0}^{T_{\mathrm{sim}}-1}\sum_{v \in V}\mathbb{E}\|1-\boldsymbol{y}(t)\|_2^2,
\end{equation}
represents the \textquotedblleft social cost\textquotedblright \ of a policy, providing an indication of the effectiveness of each policy strategy at the end of the considered horizon. Specifically, the lower the social cost in \eqref{eq:social_cost} is, the more effective the control policy is in fostering the acceptance of the target choice. Note that this quantity can be computed for both the free and policy-induced evolution of the social system.

As an actual cost-benefit analysis cannot be carried out without considering the effort required by the policymaker/stakeholder to enact it, the second indicator evaluates the policies' \textquotedblleft control cost\textquotedblright, and it is defined as
\begin{equation}
    \Delta u_{T_{\mathrm{sim}}}=\sum_{t=0}^{T_{\mathrm{sim}}-1} \|\boldsymbol{u}(t)\|_2^2.
\end{equation}

\begin{remark}[Evaluating \eqref{eq:social_cost}]\label{remark:monte_carlo} 
The cost-benefit analysis we carry out relies on $20$ Monte Carlo simulations for each value of $r$ considered, changing the realization of the (uncontrollable) external factors affecting the social system. This choice allows us to compute the index in \eqref{eq:social_cost} by averaging over different instances of the problem. 
\end{remark}
\paragraph{Results and discussion}
\begin{figure}\begin{center}
\begin{tabular}{cc}
\includegraphics[width=0.48\columnwidth]{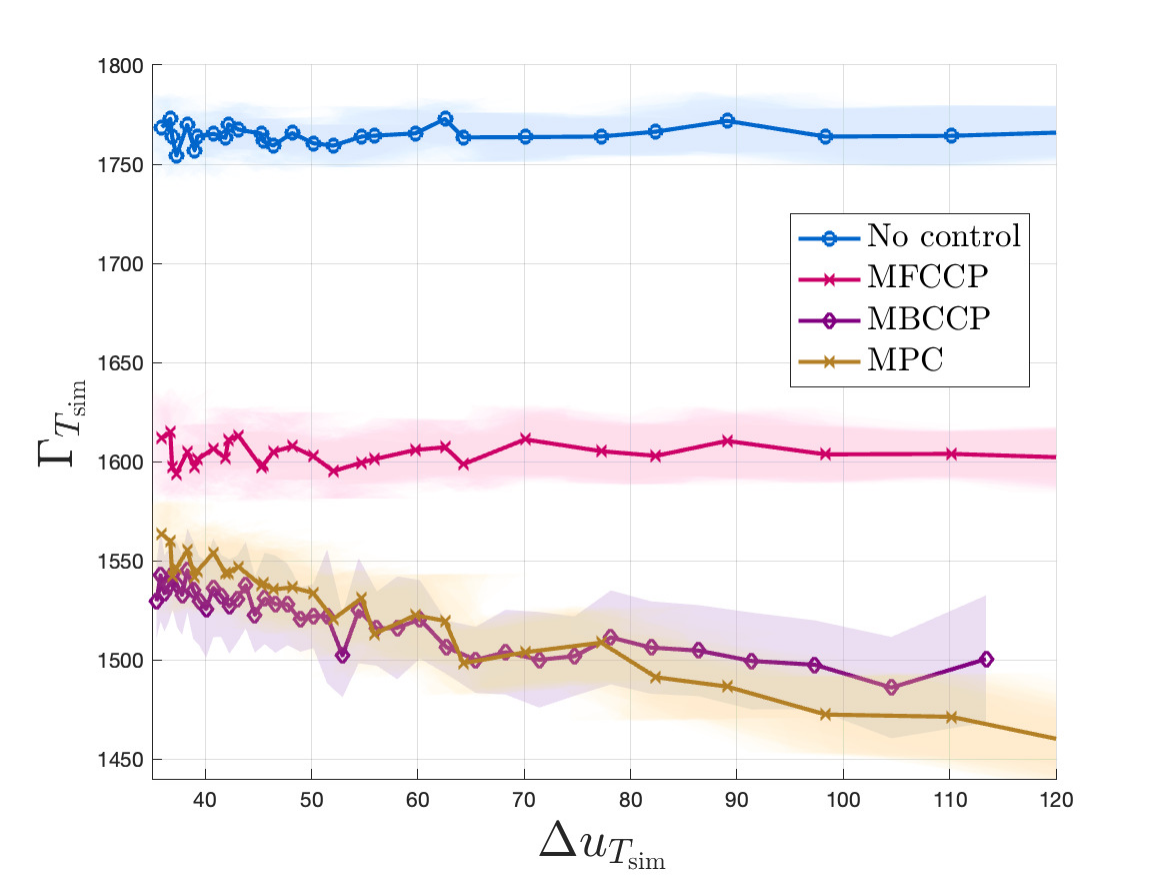}&
\includegraphics[width=0.48\columnwidth]{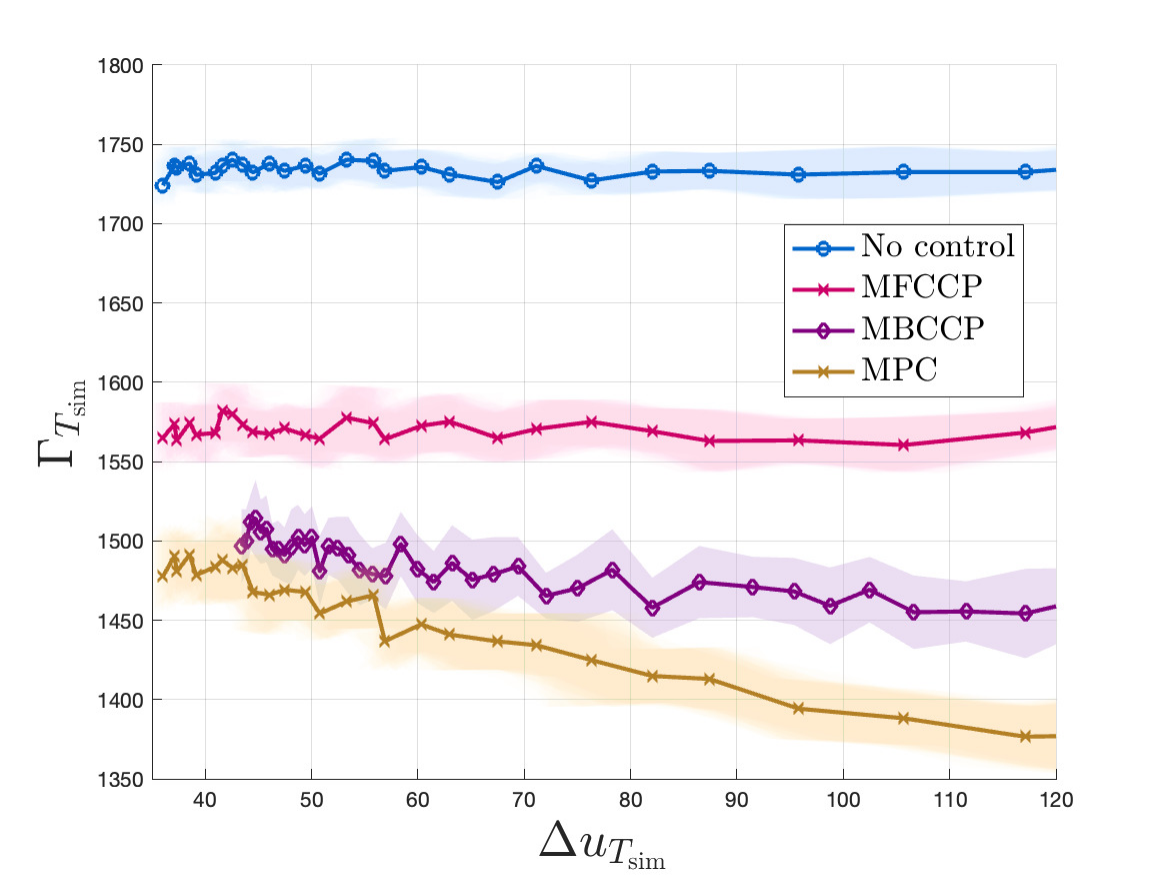}\\
(a)\footnotesize\ $\alpha=0.25$&(b)\footnotesize\ $\alpha=0.5$\\
\includegraphics[width=0.48\columnwidth]{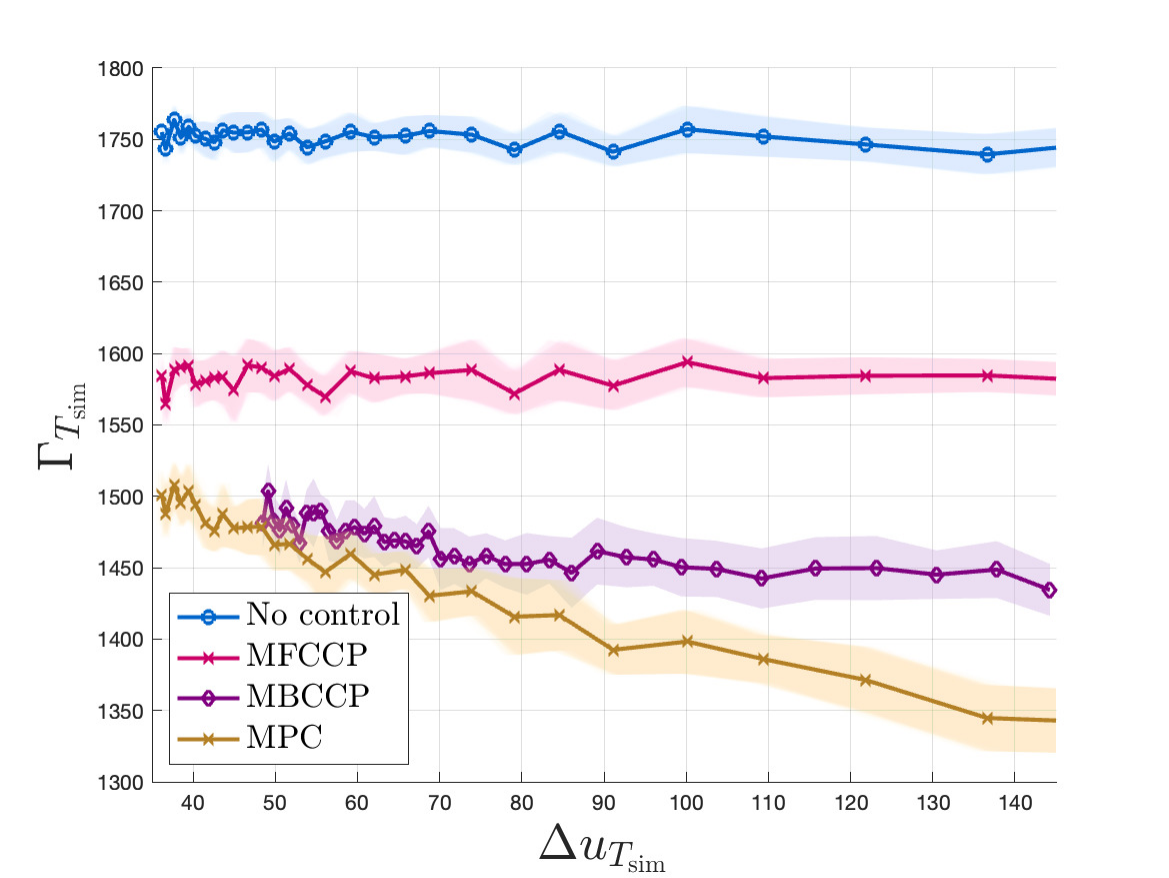}&
\includegraphics[width=0.48\columnwidth]{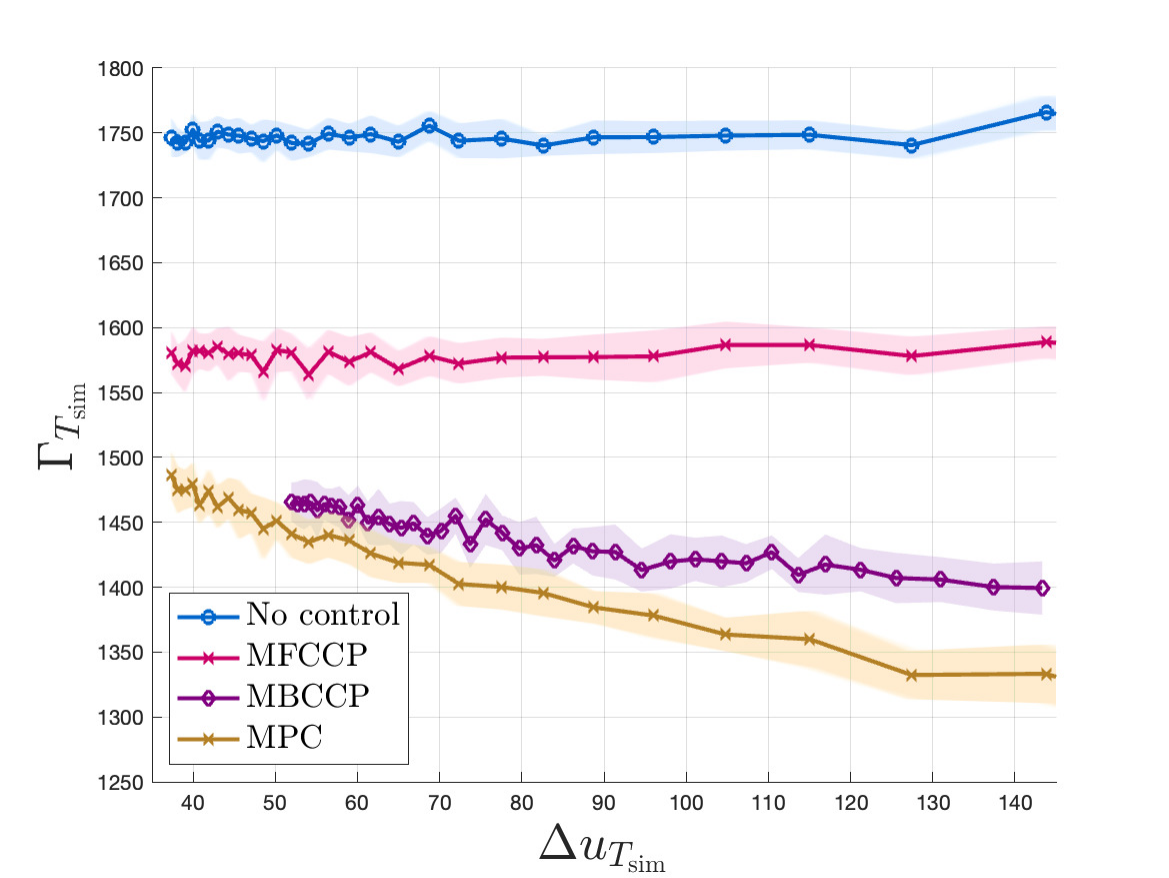}\\
(c)\footnotesize\  $\alpha=0.75$&(d)\footnotesize\ $\alpha=1$
\end{tabular}
\caption{Control cost \textit{vs} social cost for different policies for modular graphs with inter-cluster connectivity $\gamma=0.9$ for $r\in[10^-12,5]$ and different values of $\alpha$. Dotted lines indicate sample means, and shaded areas represent empirical standard deviations across 20 Monte Carlo simulations. The blue curve is used as a reference and represents the average social cost when \textit{no control} is enacted.}\label{fig:ex1}
\label{fig2:example}
\end{center}
\end{figure}
\figurename{~\ref{fig2:example}} shows the control cost against the social cost of the tested policies for the different values of the penalty parameter $r$. We stress that, although the free evolution (blue curve) and the outcome of the MFCCP strategy (magenta curve) are not affected by the choice of $r$, their value is not constant. Indeed, for each $r$, we consider $20$ new realizations of the external factors (see Remark~\ref{remark:monte_carlo}), differently affecting individual choices both in open-loop and with all the tested policies.   

As expected, irrespectively of the value taken by the time-scale separation parameter $\alpha$, enacting either non-personalized policies or personalized ones leads to a social cost that is lower than the free evolution of the system (blue curve), i.e., a more widespread acceptance of the target choice. At the same time, adopting the MFCCP strategy (magenta curve) only marginally improves the social cost with respect to the one attained with no enacted policy, as such strategy is blind to the social dynamics and solely relies on a uniform distribution of resources. In contrast, both the MBCCP (purple curve) and the MPC (yellow curve) approaches are clearly taking advantage of social imitation. Indeed, as the control cost increases, the social one decreases. This result indicates that a larger effort leads to a more widespread adoption of the target choice (and, hence, a greater social benefit if the latter aims to promote improved social welfare). Among the MBCCP strategy (see \eqref{eq:MBCCP}) and MPC-based one (see \eqref{eq:conservative_MPC2}), the latter achieves better results for the same control effort despite relying on an estimate of the hidden individual inclination. For both policies, but with a larger gain when considering the MPC-based one, it is also clear that the more frequent the social pressure (i.e., the higher $\alpha$), the more effective the enacted policy. In fact, it can be seen that as $\alpha$ increases, the social cost is lowered for the same amount of control effort.          
\subsection{Impact of the network topology}
By focusing on the MPC strategy only and by setting $r=1$, we now analyze the impact of the network topology on the \textquotedblleft relative improvement\textquotedblright \ in social cost achieved when using the MPC-based policy over adopting no policy at all, i.e.,
\begin{equation}\label{eq:relative_improvement}
\Delta_{\Gamma_{T_{\mathrm{sim}}}}=\frac{|\Gamma_{T_{\mathrm{sim}}}^{\mathrm{MPC}}-\Gamma_{T_{\mathrm{sim}}}^{\mathrm{ol}}|}{\Gamma_{T_{\mathrm{sim}}}^{\mathrm{ol}}},
\end{equation}
where $\Gamma_{T_{\mathrm{sim}}}^{\mathrm{MPC}}$ and $\Gamma_{T_{\mathrm{sim}}}^{\mathrm{ol}}$ are the values of the social cost defined in \eqref{eq:social_cost} achieved with the MPC policy and in open-loop, respectively. To make this assessment, we consider new randomly generated modular graphs with $100$ nodes (as before), varying the inter-cluster connectivity parameter $\gamma$. This choice allows us to compare the outcome of the policy for both homogeneous networks (i.e., for $\gamma$ small) and clustered ones (corresponding to large values $\gamma$). Note that, for each value of $\gamma$, we consider $50$ different realizations of the modular graph. 

In contrast with the previous analysis, we fix $\boldsymbol{\Lambda}$ by imposing $\lambda_{v}=0.9$ for all $v \in \mathcal{V}$ and we set $x_{v}(0)=0.1$ for all $v \in \mathcal{V}$. This choice enables us to avoid confounding effects induced by polarization of the population as well as differences in the impact of social interactions and external factors in the opinions' evolution.  

\figurename{~\ref{fig3:example}} reports both a scatter plot and a box plot of the relative improvement in social cost (see \eqref{eq:relative_improvement}), considering both $\alpha = 0.25$ and $\alpha = 0.75$. It can be observed that when the time-scale separation parameter $\alpha$ is lower and, thus, social pressure is low, the network topology seems not to play a well-defined role in the outcome of the policy. Indeed, the best results are obtained when the inter-cluster connectivity has intermediate values. On the contrary, under higher social pressure ($\alpha = 0.75$), greater social impacts are obtained for higher values of connectivity. This result ultimately indicates, as expected, that the MPC-based policy benefits from the indirect effect of social contagion. 

\section{Concluding remarks}\label{sec:conclusions}
This work presented a novel extension of the Friedkin and Johnsen model that characterizes the evolution of individual inclinations subject to social influence and to the effects of unpredictable phenomena and external (controlled) interventions. Based on this model and the analysis of its properties, we have devised a set of policy design strategies, also looking at the optimal balance between fostering the social acceptance on average of a target action and containing costs. Our numerical analysis spotlighted the differences, benefits, and drawbacks of the proposed strategies, setting the ground for further research into the use of these models for effective policy design in a social setting. Another avenue of investigation will be assessing the realism of our modeling choices through real-world data.
\begin{figure}[!tb]
\begin{center}
\begin{tabular}{cc}
\includegraphics[width=0.46\columnwidth]{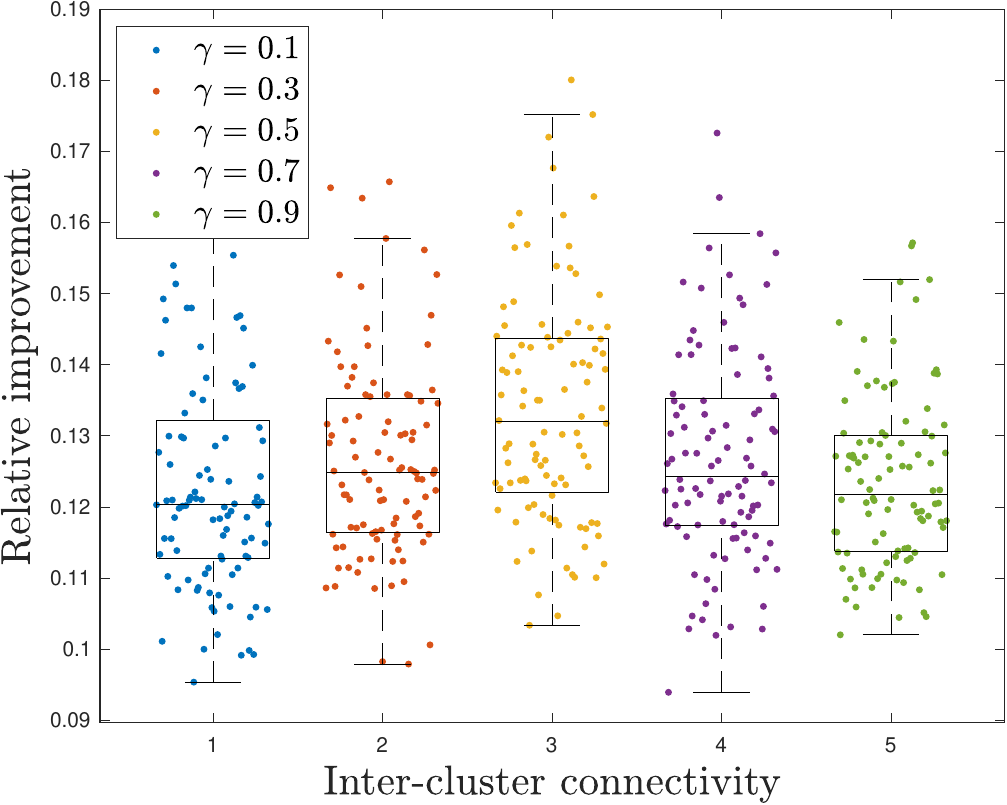}&
\includegraphics[width=0.46\columnwidth]{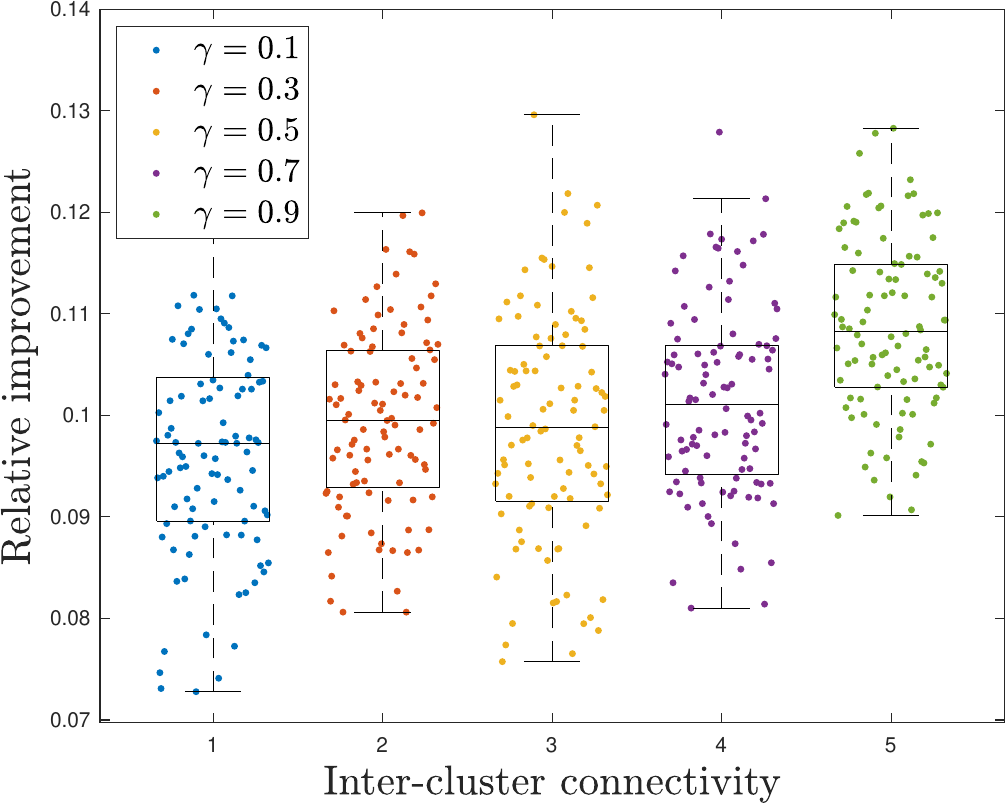}\\
(a)\footnotesize\ $\alpha=0.25$&(b)\footnotesize\ $\alpha=0.75$
\end{tabular}
\caption{Box and scatter plot showing the relative shift of social cost for $50$ realizations of random modular graphs with $100$ nodes for each value of the inter-cluster connectivity parameter $\gamma$ for different values of $\alpha$.}\label{fig:ex1}
\label{fig3:example}
\end{center}
\end{figure}

\bibliographystyle{plain}
\bibliography{Arxiv_tac_25.bib}

\end{document}